\documentclass[]{scrartcl}
\usepackage{amsmath,amssymb,amsfonts}
\usepackage{amsthm}
\usepackage{algorithmic}
\usepackage{graphicx}
\usepackage{textcomp}
\usepackage{xcolor}
\usepackage[style=ieee,backend=biber,url=false,doi=false,isbn=false,date=year,eprint = true, citestyle=numeric-comp,maxbibnames=4,minbibnames=4,maxcitenames=4,mincitenames=4]{biblatex}
\usepackage{booktabs}
\usepackage{dsfont}
\usepackage[hidelinks]{hyperref}
\usepackage{subcaption}

\usepackage{multirow}

\bibliography{alexLit}

\newtheorem{remark}{Remark}

\newcommand{\Fp}[1]{{\color{black} #1}}
\newcommand{\Ss}[1]{{\color{black} #1}}

\newcommand{\Aee}[1]{{\color{red} }}

\usepackage[ ruled]{algorithm2e}

\begin{document}
\title{Scalable Primal Decomposition Schemes for Large-Scale Infrastructure Networks}
\author{Alexander Engelmann, Sungho Shin, \\ François Pacaud, and Victor M. Zavala
\thanks{This work was supported by the U.S. Department of Energy, Office of Science, Advanced Scientific Computing Research, under contract number DE-AC02-06CH11357. VM Zavala also acknowledges partial support from the U.S. National Science Foundation under award CBET-2315963.
This research was conducted while AE  was with TU Dortmund University, Dortmund, Germany (alexander.engelmann@ieee.org). SS is with Massachusetts Institute of Technology, Cambridge, MA, USA (sushin@mit.edu). FP is with Centre Mathématiques et Systèmes, Mines Paris-PSL, Paris, France (francois.pacaud@minesparis.psl.eu). VZ is with University of Wisconsin-Madison, Madison, WI, USA and Argonne National Laboratory, Lemont, IL, USA (victor.zavala@wisc.edu).}}

\newtheorem{ass}{Assumption}
\newtheorem{lem}{Lemma}
\newtheorem{thm}{Theorem}

  \renewcommand{\algorithmautorefname}{Algorithm}
  \renewcommand{\appendixautorefname}{Appendix}

\setlength{\textfloatsep}{5pt}

\makeatletter
\newcommand{\thickhline}{%
	\noalign {\ifnum 0=`}\fi \hrule height 1pt
	\futurelet \reserved@a \@xhline
}
\makeatother

\maketitle

\begin{abstract}
  The  operation of large-scale infrastructure networks requires  scalable optimization schemes. 
  To guarantee safe system operation, a high degree of feasibility in a small number of iterations is important. Decomposition schemes can help to achieve scalability. In terms of feasibility, however,  classical  approaches such as the alternating direction method of multipliers (ADMM)  often converge slowly. In this work, we present  primal decomposition schemes for hierarchically structured strongly convex QPs.
  These schemes offer high degrees of feasibility in a small number of iterations in combination with global convergence guarantees. We benchmark their performance against the centralized off-the-shelf interior-point solver Ipopt and  ADMM on problems with up to 300,000 decision variables and constraints. We find that the proposed approaches solve problems as fast as Ipopt, but with reduced communication and without requiring a full model exchange. Moreover, the proposed schemes achieve a higher accuracy than ADMM.
\end{abstract}

\section{Introduction} \label{sec:intro}

The operation of  infrastructure networks such as power systems, district heating grids or gas networks is challenging. In many cases, these networks are large and composed of many complex subsystems such as lower-level networks or buildings. Operation  is often  based on numerical optimization  due to its flexibility  and recent advances in solver development, which allows to solve  large-scale problems quickly and to a high accuracy. For large networks, however, a centralized solution is often not desirable since,  a), the problem becomes computationally challenging, even with state-of-the-art solvers; b), information collection in a central entity should be avoided due to confidentiality and privacy concerns, and, c), the responsibility for operation and updates in modeling should stay mainly in the subsystems.

One line of research addresses the above challenges via aggregation.
Here, the idea is to simplify the subproblems by projecting the constraint set on the coupling variables of the infrastructure network.
Examples for this can be found for power systems~\cite{Capitanescu2018,Kalantar-Neyestanaki2020}.
A drawback of this approach is a loss of optimality.
Moreover, aggregation is often not straightforward, feasibility is hard to guarantee and disaggregation requires solving additional local optimization problems.

A second line of research is based on distributed optimization. Prominent approaches are primal and dual first-order algorithms such as Lagrangian dual decomposition, the Alternating Direction Method of Multipliers~(ADMM) \cite{Everett1963,Boyd2011}, and primal (sub)-gradient-based schemes \cite{Nedic2017,Ryu2022}. Application examples range from the operation of power systems \cite{Erseghe2014,Kim2000}, over gas networks~\cite{Shin2021}, district heating systems \cite{Huang2017,Cao2019}, to water networks~\cite{Coulbeck1988}.
With their at most linear rate of convergence, these approaches  often require many iterations to converge even for a modest solution quality.
This is often prohibitive for real-time implementation.

Distributed second-order methods exhibit faster convergence.
Here, classical approaches aim at decomposing the block-structure of the Karush-Kuhn-Tucker~(KKT) system within interior-point algorithms \cite{Chiang2014,Zavala2008a} or sequential quadratic programming \cite{Varvarezos1994}.
Alternative second-order methods based on augmented Lagrangians can be found in \cite{Engelmann2019c,Houska2016}. These approaches typically require an expensive central coordination, although it is possible to partially alleviate computation by decentralizing the Newton-steps~\cite{Engelmann2020b,Engelmann2021a, Stomberg2022a}.

Primal decomposition schemes come with the advantage  a high degree of feasibility and optimality in a small number of iterations \cite{Geoffrion1970,DeMiguel2006,DeMiguel2008}.
	For achieving this, they  require a hierarchical problem structure, i.e. a star  as the underlying graph.
In this sense, they are more restrictive than the aforementioned approaches.
In infrastructure networks hierarchical problem structures are common, however.
The main idea of primal decomposition is to  construct lower-level problems coordinated by one upper-level problem, where the upper-level problem considers the lower-level problems  by their optimal value functions.
Primal decomposition has been very successful in solving large-scale problems from chemical engineering \cite{Zavala2008,Yoshio2021} and some of the largest Quadratic Programs~(QPs) and Nonlinear Programs (NLPs) from power systems \cite{Tu2021,Petra2021,Curtis2021}. Moreover, primal decomposition  allows to use specialized, domain-specific solvers to solve the subproblems and the master problem efficiently \cite{DeMiguel2006}.

In this work, we propose two primal decomposition schemes for solving large-scale strongly convex QPs, with global convergence guarantees.
Both methods rely respectively on augmented Lagrangians and exact  $\ell 1$-penalties for ensuring feasibility in the subproblems.
Similar $\ell1$-penalty based approaches have been proposed in previous works \cite{DeMiguel2006,Tu2021}.
In contrast to \cite{Tu2021}, our work is not restricted to a specific application and can be used on any strongly convex hierarchically structured QP.
The augmented-Lagrangian framework is new to the best of our knowledge.
We show that the augmented Lagrangian formulation exhibits improved performance compared to the $\ell$1 formulation. Moreover, we demonstrate that the algorithms are faster than off-the-shelf interior-point solvers.
We benchmark our algorithms against a distributed ADMM and the nonlinear solver \texttt{Ipopt}.
As benchmarks, we consider the operation of HVAC systems in a city district with a variable number of buildings and with up to $300,000$ decision variables and inequality constraints and two Optimal Power Flow problems with up to 7,852 buses.

\subsubsection*{Notation}
Given $A\in \mathbb{R}^{m\times n}$, $[A]_j$ denotes the $j$th row of $A$ and $\operatorname{nr}(A)\doteq  m$ corresponds to the number of rows of $A$.
The Lagrange multiplier $\lambda\in \mathbb{R}^{n_g}$ associated to a constraint $g : \mathbb{R}^{n_x} \rightarrow \mathbb R^{n_g}$ is written as  $g(x) = 0 \;|\; \lambda$.
Given a vector $v\in \mathbb R^n$, $D= \operatorname{diag}(v) \in \mathbb{R}^{n\times n}$ denotes a matrix with the elements of $v$ on the main diagonal.
For a tuple of matrices $(A,B)$,
$\operatorname{blkdiag}(A,B)$ denotes their block-diagonal concatenation.

\section{Problem Formulation} \label{sec:probForm}

\begin{figure}
	\centering
	\includegraphics[width=.4\linewidth]{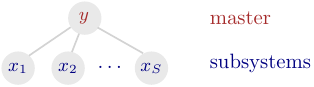}
	\caption{Star graph of problem \eqref{eq:sepQP}.}
	\label{fig:stargraph}
\end{figure}

Many infrastructure network problems can be formulated as strongly convex QPs over a set of subsystems $\mathcal S = \{1,\dots,S\}$,
\begin{subequations}\label{eq:sepQP}
	\begin{align}
		\min_{\{x_i\}_{i\in \mathcal S},y}\sum_{i \in \mathcal S}
		\frac{1}{2}
		\begin{bmatrix}
			x_i \\ y
		\end{bmatrix}^	{\hspace{-.2em}\top}
		\hspace{-.2em}
		\begin{bmatrix}
			H_i^{xx} & H_i^{xy} \\
			H_i^{xy\top} & H_i^{yy}
		\end{bmatrix}
		\begin{bmatrix}
			x_i \\ y
		\end{bmatrix}&
	\hspace{-.2em}
		+
			\hspace{-.2em}
		\begin{bmatrix}
			h_i^x \\ h_i^y
		\end{bmatrix}^	{\hspace{-.2em}\top}
		\hspace{-.2em}
		\begin{bmatrix}
			x_i \\ y
		\end{bmatrix} \label{eq:sepCost}
		\\
		\text{subject to } \;
		\begin{bmatrix}
			A_i^x & A_i^y
		\end{bmatrix}
		\begin{bmatrix}
			x_i^\top & y^\top
		\end{bmatrix}^\top - b_i &= 0,\;\;  i \in \mathcal S, \label{eq:localEq}\\
		\begin{bmatrix}
			B_i^x & B_i^y
		\end{bmatrix}
		\begin{bmatrix}
			x_i^\top & y^\top
		\end{bmatrix}^\top -d_i&\leq 0,\;\;  i \in \mathcal S, \label{eq:localIneq}\\
		A^y y -b^y = 0, \quad 		B^y y -d^y &\leq  0\label{eq:GlobalIneq}.
	\end{align}
\end{subequations}
Here, the global decision  vector $x=[x_1,\dots,x_S]^\top $ is composed of local decision variables $x_i \in \mathbb R^{n_{xi}}$, where each $x_i$ belongs to one subsystem $i \in \mathcal S$.
The decision variables  $y\in \mathbb R^{n_y}$  are ``global'' in the sense that they belong to the interconnecting infrastructure network, described by the constraints \eqref{eq:GlobalIneq}.
Each coefficient matrix/vector in  the objective \eqref{eq:sepCost} and the constraints \eqref{eq:localEq}, \eqref{eq:localIneq} belongs to one  $i \in \mathcal S$.

Observe that problem \eqref{eq:sepQP} is defined over a star graph, where $y$ and constraint \eqref{eq:GlobalIneq} correspond to the root vertex, and $x_i$ and constraints \eqref{eq:localEq}, \eqref{eq:localIneq} belong/couple the root vertex to all leafs (\autoref{fig:stargraph}).
This structure is common in many infrastructure networks such as electricity grids, gas networks or district heating systems, which are composed of a network as the root and complex subsystems such as households, distribution grids or industrial facilities as leafs \cite{Erseghe2014,Kim2000,Shin2021,Huang2017,Cao2019,Coulbeck1988}.
These applications often require a high degree of feasibility in a small number of iterations without full model exchange. The main objective of this work is to develop primal decomposition schemes able to achieve that goal.


\section{Primal Decomposition Schemes} \label{sec:primDec}
In contrast to  duality-based techniques such as ADMM or  dual decomposition,
primal decomposition decomposes  entirely in the primal space, i.e. no dual variables are updated in the solution process.
The main idea here is to replace the subproblems in \eqref{eq:sepQP} by their optimal value functions.
Specifically, one reformulates \eqref{eq:sepQP} as
\begin{align} \label{eq:valFunReform}
	\min_{y} \sum_{i \in \mathcal S} \;\phi_i(y), \;\; \text{ s.t. }\;\; 	A^y y -b^y = 0,\;
	B^y y -d^y \leq  0,
\end{align}
where for all $i \in \mathcal S$, the {value function} $\phi_i$ is defined as
\begin{subequations} \label{eq:ValueFun}
	\begin{align}
		\phi_i(y)\hspace{-.1em} \doteq \hspace{-.1em}	\min_{x_i}
		\frac{1}{2}
		\hspace{-.1em}
		\begin{bmatrix}
			x_i \\ y
		\end{bmatrix}^\top
	\hspace{-.3em}
		\begin{bmatrix}
			H_i^{xx} & H_i^{xy} \\
			H_i^{xy\top} & H_i^{yy}
		\end{bmatrix}
	\hspace{-.1em}
		\begin{bmatrix}
			x_i \\ y
		\end{bmatrix}&
	\hspace{-.2em}
		+
		\hspace{-.2em}
		\begin{bmatrix}
			h_i^x \\ h_i^y
		\end{bmatrix}^\top
	\hspace{-.3em}
		\begin{bmatrix}
			x_i \\ y
		\end{bmatrix}
		\\
		\text{subject to } \;
		\begin{bmatrix}
			A_i^x & A_i^y
		\end{bmatrix}
		\begin{bmatrix}
			x_i^\top & y^\top
		\end{bmatrix}^\top - b_i &= 0,\;  i \in \mathcal S, \\
		\begin{bmatrix}
			B_i^x & B_i^y
		\end{bmatrix}
		\begin{bmatrix}
			x_i^\top & y^\top
		\end{bmatrix}^\top -d_i&\leq 0,\;  i \in \mathcal S.  \label{eq:ineqCstr}
	\end{align}
\end{subequations}
The key idea is to apply standard algorithms for solving~\eqref{eq:valFunReform}
by optimizing only with respect to the coupling variables~$y$.
Doing so can lead to enhanced robustness, as the complexity of the subproblems is not exposed to the algorithm solving~\eqref{eq:valFunReform}.

Algorithms for solving \eqref{eq:valFunReform} typically require first-order and possibly second-order derivatives of all $\{\phi_i\}_{i \in \mathcal S}$.
Since all $\{\phi_i\}_{i \in \mathcal S}$ are non-smooth because of the inequality constraints, one typically relies on smooth reformulations.
Inspired by interior-point methods~\cite{DeMiguel2006}, we introduce log-barrier functions and
slack variables $s_i\in \mathbb R^{n_{ii}}$, which approximate \eqref{eq:ValueFun} by
\begin{subequations} \label{eq:ValueFunBarr}
	\begin{align}
		\Phi_i^\delta(y) \doteq	\min_{x_i,s_i}
		\frac{1}{2}
		\begin{bmatrix}
			x_i \\ y
		\end{bmatrix}^\top
		\begin{bmatrix}
			H_i^{xx} & H_i^{xy} \\
			H_i^{xy\top} & H_i^{yy}
		\end{bmatrix}
		\begin{bmatrix}
			x_i \\ y
		\end{bmatrix} \label{eq:barrObj}
		+ 
		\begin{bmatrix}
			h_i^x \\ h_i^y
		\end{bmatrix}^\top
		\begin{bmatrix}
			x_i \\ y
		\end{bmatrix}
		 - &\delta \mathds 1^\top \ln(s_i)\\
		\text{subject to } \;
		\begin{bmatrix}
			A_i^x & A_i^y
		\end{bmatrix}
		\begin{bmatrix}
			x_i^\top & y^\top
		\end{bmatrix}^\top - b_i = 0&,\;  i \in \mathcal S, \label{eq:BarrEConstr} \\
		\begin{bmatrix}
			B_i^x & B_i^y
		\end{bmatrix}
		\begin{bmatrix}
			x_i^\top & y^\top
		\end{bmatrix}^\top -d_i + s_i= 0&,\;  i \in \mathcal S, \label{eq:BarrIConstr}
	\end{align}
\end{subequations}
where $\delta \in \mathbb R_+$ is a barrier parameter, $\mathds 1 \doteq [1,\dots,1]^\top$, and the $\ln(\cdot)$ is evaluated component-wise. Note that  $\lim_{\delta \rightarrow 0}\Phi_i^\delta(y) = \phi_i(y)$, and that $\Phi_i^\delta$ is smooth\footnote{Under standard regularity assumptions \cite[A1-C1]{DeMiguel2008}.}. A basic primal decomposition strategy with smoothing is summarized in \autoref{alg:basPrimDec}.


\begin{algorithm}[t]
	\SetAlgoLined
	\caption{A basic primal decomposition scheme.} \label{alg:basPrimDec}
	\BlankLine
	\textbf{Initialize} $y^0, \delta^0$. \\
	\While{\textnormal{not terminated}}{
		1) Solve \eqref{eq:valFunReform} for $\phi_i \equiv \Phi_i^\delta $ with a NLP solver; in case the NLP solver calls $(\nabla_y \Phi_i^\delta, \nabla^{2}_y \Phi_i^\delta)$, compute them  locally for all  $i \in \mathcal S$. \\
		2) Decrease $\delta$.}
	\textbf{Return} $y^k$, $\{x_i^k\}_{i\in \mathcal S}$.
	\BlankLine
\end{algorithm}

\subsection*{Dealing With Infeasibility}
An issue in \autoref{alg:basPrimDec} is that  the subproblems \eqref{eq:ValueFun} may be infeasible for a given  $y$. One way of circumventing this is to introduce auxiliary variables $z_i \in \mathbb{R}^{n_y}$ and to use relaxation techniques either based on augmented Lagrangians or on exact $\ell$1-penalties. Consider a set of auxiliary  variables $\{z_i\}_{i\in \mathcal S}$ and introduce additional constraints $z_i = y$ for all $i \in \mathcal S$. Then, we reformulate \eqref{eq:ValueFunBarr} equivalently  \footnote{Observe that we have replaced $y$ by $z_i$ in the constraints here but not in the objective.
Exchanging $y$ in the objective is possible but might lead to a different numerical behavior.}
\begin{subequations} \label{eq:ValueFunBarr2}
	\begin{align}
		\Phi_i^\delta(y) =	\min_{x_i,s_i,z_i}
		\frac{1}{2}
		\begin{bmatrix}
			x_i \\ y
		\end{bmatrix}^\top
		\begin{bmatrix}
			H_i^{xx} & H_i^{xy} \\
			H_i^{xy\top} & H_i^{yy}
		\end{bmatrix}
		\begin{bmatrix}
			x_i \\ y
		\end{bmatrix}
		+
		\begin{bmatrix}
			h_i^x \\ h_i^y
		\end{bmatrix}^\top
		\begin{bmatrix}
			x_i \\ y
		\end{bmatrix}
		- &\delta \mathds 1^\top \ln(s_i)\\
		\text{subject to } \;
		\begin{bmatrix}
			A_i^x & A_i^y
		\end{bmatrix}
		\begin{bmatrix}
			x_i^\top & z_i^\top
		\end{bmatrix}^\top - b_i = 0&,\;  i \in \mathcal S, \\
		\begin{bmatrix}
			B_i^x & B_i^y
		\end{bmatrix}
		\begin{bmatrix}
			x_i^\top & z_i^\top
		\end{bmatrix}^\top -d_i + s_i= 0&,\;  i \in \mathcal S,\\
	z_i=y&,\; i \in \mathcal S \label{eq:consConstr},
	\end{align}
\end{subequations}
which can still be infeasible, but paves the way for
augmented Lagrangian and exact $\ell$1 relaxations.

\subsubsection*{Augmented Lagrangian Relaxation}

\begin{algorithm}[t]
	\SetAlgoLined
	\caption{AL-based primal decomposition.} \label{alg:ALPrimDec}
	\BlankLine
	\textbf{Initialize} $y^0, \delta^0$, $\rho^0$; $\lambda_i=0, i\in \mathcal{S}$. \\
	\While{\textnormal{phase 1}}{
		1) Solve the master problem \eqref{eq:valFunReform} for $\phi_i \equiv 	\Phi_i^{\delta,\rho} $ with a NLP solver; in case the NLP solver calls  $\phi_i(y^k),\nabla_y \phi_i(y^k)$, or  $\nabla_{yy}^2 \phi_i(y^k)$,  broadcast $y^k$ to all $i \in \mathcal S$ and compute them  locally. \\
		2) Decrease $\delta$, increase $\rho$.}
	\While{\textnormal{phase 2}}{
		3) Solve \eqref{eq:valFunReform} as in phase 1.\\
		4)  Broadcast $y^k$ to all $i \in \mathcal S$ and update $\lambda_i^k$ by \eqref{eq:lamUpdate}.}
	\BlankLine
	\textbf{Return} $y^k$, $\{x_i^k\}_{i\in \mathcal S}$. \label{alg:ALprimDec}
\end{algorithm}

A simple way of making~\eqref{eq:ValueFunBarr2} feasible for all $y$ is to relax \eqref{eq:consConstr} via a quadratic penalty \cite{DeMiguel2006}. However, in this case, large penalty parameters might lead to numerical difficulties and feasibility can in general not be guaranteed for a finite penalty parameter. Hence, we use an Augmented Lagrangian~(AL) approach to solve \eqref{eq:valFunReform} for late outer iterations with a constant barrier parameter $\delta$.
Assigning  the Lagrange multiplier $\lambda_i$ to \eqref{eq:consConstr},
we relax~\eqref{eq:consConstr} in an Augmented Lagrangian fashion by adding the terms $\lambda_i^{k\top} (y-z_i) + \frac{\rho}{2}\|y-z_i\|_2^2$ to the objective:
	\begin{align} \label{eq:subProbBarrQP2}
\hspace{-.4em}		\Phi_{i}^{\delta,\rho}(y,\lambda^k_i) \hspace{-.1em}\doteq \hspace{-.5em} \min_{x_i,s_i,z_i}\hspace{-.1em} \frac{1}{2}  \hspace{-.3em}
		\begin{bmatrix}
			x_i \\  y \\z_i
		\end{bmatrix}^{\hspace{-.4em}\top} \notag 
	\hspace{-.4em}
		\begin{bmatrix}
			H_i^{xx} & H_i^{xy} &0 \\
			H_i^{xy\top} & \hspace{-.6em}H_i^{yy} + \rho I\hspace{-.6em}& -\rho I \\
			0 & -\rho I & \rho I
		\end{bmatrix}
		\hspace{-.4em}
		\begin{bmatrix}
			x_i \\ y \\z_i
		\end{bmatrix} &
		+
		\begin{bmatrix}
			h_i^x \\ h_i^y + \lambda_i^k \\ -\lambda^k_i
		\end{bmatrix}^\top
		\begin{bmatrix}
			x_i \\  y \\z_i
		\end{bmatrix} - \delta \mathds 1^\top \ln (s_i) \\
			\text{s.t. }\;
		\begin{bmatrix}
			A_i^x & A_i^y
		\end{bmatrix}
		\begin{bmatrix}
			x_i^\top & z_i^\top
		\end{bmatrix}^\top
		- b_i= 0&,  \; |\;  \;\gamma_i  \\
		\begin{bmatrix}
			B_i^x & B_i^y
		\end{bmatrix}
		\begin{bmatrix}
			x_i^\top & z_i^\top
		\end{bmatrix}^\top + s_i -d_i = 0&,  \; |\;  \;\mu_i. \notag 
\end{align}
Here, $(\gamma_i,\mu_i)$ are Lagrange multipliers corresponding to the constraints in the same line.
Note that  the solution of  \eqref{eq:subProbBarrQP2}  can be forced to be arbitrarily close to that of  \eqref{eq:ValueFunBarr2} by letting $\rho \rightarrow \infty$, if \eqref{eq:subProbBarrQP2} is feasible for  $y$.
In addition, if one has a good Lagrange multiplier estimate $\lambda^k \approx \lambda^\star$, \eqref{eq:ValueFunBarr2} and \eqref{eq:subProbBarrQP2} are equivalent for a finite $\rho < \infty$ \cite[Sec. 3.2.1]{Bertsekas1999}, \cite[Thm. 17.5]{Nocedal2006}.
We will exploit this fact in the following.

\begin{algorithm}[t]
	\SetAlgoLined
	\caption{$\ell1$-based primal decomposition.} \label{alg:slPrimDec}
	\BlankLine
	\textbf{Initialize} $y^0, \delta^0, \bar \lambda$ large enough. \\
	\While{\textnormal{not terminated}}{
		1) Solve the master problem \eqref{eq:valFunReform} for $\phi_i \equiv \Phi_i^{\delta,\bar \lambda} $ with a NLP solver; in case the NLP solver calls $\phi_i(y^k),\nabla_y \phi_i(y^k)$, or  $\nabla_{yy}^2 \phi_i(y^k)$,  broadcast $y^k$ to all $i \in \mathcal S$ and compute them  locally. \\
		2) Decrease $\delta$.}
	\textbf{Return} $y^k$, $\{x_i^k\}_{i\in \mathcal S}$. \label{alg:SL1primDec}
	\BlankLine
\end{algorithm}

Primal decomposition based on the augmented Lagrangian works as follows: In phase~1, the barrier parameter $\delta$ and the penalty parameter $\rho$ are increased/decreased simultaneously.
In phase~2, when $\delta/\rho$ are sufficiently small/large, both are held constant and a standard augmented Lagrangian algorithm is applied to the resulting optimization problem to obtain feasibility in \eqref{eq:consConstr}.
We use the standard first-order update rule from augmented Lagrangian algorithms \cite[Chap. 17.3]{Nocedal2006}
\begin{align} \label{eq:lamUpdate}
	\lambda_i^{k+1} = \lambda_i^k + \rho(y^k-z_i^k).
\end{align}
The resulting scheme is summarized in \autoref{alg:ALPrimDec}.

\subsubsection*{$\ell1$-penalty relaxation}

A second variant to ensure feasibility is to relax \eqref{eq:consConstr} via an $\ell1$-penalty function.
This has the advantage that it is exact also for a finite penalty parameter $\bar \lambda \in \mathbb{R}_+$ without the need for Lagrange-multiplier estimation. By doing so, the objective becomes non-smooth.
However, the non-smoothness can be eliminated by using an elastic relaxation \cite[p.535]{Nocedal2006}:
the  $\ell1$-penalty $\min_{y, z_i} \|y - z_i \|_1$ is reformulated
by introducing two non-negative auxiliary variables $v_i, w_i \in \mathbb{R}^{n_y}_+$
as
$\min_{y, z_i, v_i, w_i} v_i + w_i$ subject to $y - z_i = v_i - w_i$.
The corresponding reformulation of \eqref{eq:ValueFunBarr2} reads
	\begin{align} \label{eq:subProbBarrQP3}
	\Phi_i^{\delta,\bar \lambda}\hspace{-.1em}(y) \hspace{-.0em} \doteq& \hspace{-.0em}	\min_{\substack{x_i,s_i,z_i \\ v_i,w_i}}\hspace{-.0em}
		\frac{1}{2}
		\hspace{-.0em}
		\begin{bmatrix}
			x_i \\ z_i
		\end{bmatrix}^{\hspace{-.3em}\top}
		\hspace{-.5em}
		\begin{bmatrix}
			H_i^{xx} &\hspace{-.5em} H_i^{xy} \\
			H_i^{xy\top} &\hspace{-.5em} H_i^{yy}
		\end{bmatrix}
	\hspace{-.3em}
		\begin{bmatrix}
			x_i \\ z_i
		\end{bmatrix}
		\hspace{-.1em}+ 		\hspace{-.1em}
		\begin{bmatrix}
			h_i^x \\ h_i^y
		\end{bmatrix}^{\hspace{-.3em}\top}
			\hspace{-.5em}
		\begin{bmatrix}
			x_i \\ z_i
		\end{bmatrix} \notag
		  \\
		& \qquad \qquad \qquad \qquad + \bar \lambda   \mathds 1^{\hspace{-.1em}\top} \hspace{-.1em}(v_i+w_i)   \hspace{-.0em}-\hspace{-.0em}\delta  (\mathds 1^{\hspace{-.1em}\top} \ln(s_i) \hspace{-.1em} + \hspace{-.1em}\mathds 1^{\hspace{-.1em}\top} \ln(v_i)\hspace{-.1em} +\hspace{-.1em}\mathds 1^{\hspace{-.1em}\top} \ln(w_i) ) \\
	&
	\begin{aligned}
			\text{s.t. }
		\begin{bmatrix}
			A_i^x & A_i^y
		\end{bmatrix}
		\begin{bmatrix}
			x_i^\top & z_i^\top
		\end{bmatrix}^\top
		- b_i&= 0,  &|& \;\gamma_i,\\
		\begin{bmatrix}
			B_i^x & B_i^y
		\end{bmatrix}´
		\begin{bmatrix}
			x_i^\top & z_i^\top
		\end{bmatrix}^\top + s_i -d_i &= 0,  	&|& \;\mu_i, \\
		y-z_i - v_i + w_i &= 0,   \hspace{-.5em}&|& \;\chi_i, \notag
	\end{aligned}
	\end{align}
where the bounds $(v_i,w_i) \geq 0$ are \Ss{replaced by} log-barrier functions.
If one chooses $\bar \lambda < \infty$ large enough,
\eqref{eq:ValueFunBarr2} and \eqref{eq:subProbBarrQP3} are equivalent \cite[Thm 17.3]{Nocedal2006}.

Primal decomposition based on the $\ell1$-penalty solves \eqref{eq:valFunReform} using $	\Phi_i^{\delta,\bar \lambda}$ with a fixed $\bar \lambda$ larger than a certain threshold and decreases the barrier parameter $\delta$ during the iterations.
The overall algorithm is summarized in \autoref{alg:slPrimDec}.

	The communication of Algorithm~1 and 2 is illustrated in \autoref{fig:iter}. Communication involves $\phi_i(y)$, its derivatives,  and $\lambda_i$. Hence,  the dimension of information exchange depends only on the dimension of coupling variables  $n_y$.

\begin{figure}[t]
	\centering
	\includegraphics[width=.55\linewidth]{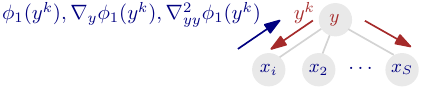}
	\caption{Communication in Algorithms~\ref{alg:ALPrimDec} and \ref{alg:SL1primDec}.}
	\label{fig:iter}
\end{figure}

\section{Computing sensitivities} \label{sec:compSens}
 Next, we review how to compute $\nabla_y \Phi_i^\delta$ and $\nabla_{yy}^2 \Phi_i^\delta$ under standard regularity assumptions based on the  implicit function theorem  \cite{DeMiguel2008}.
 Reformulate \eqref{eq:ValueFunBarr} by
	\begin{align} \label{eq:barrProbRewrite}
	 \Phi_i^\delta(y)&=\min_{x_i,s_i}f_i^\delta(x_i,s_i;y) \\
	  \;\text{subject to}\qquad  g_i(x_i&;y)=0\;\;|\;\; \gamma_i,\qquad  h_i(x_i;y) + s_i =0\;\;|\;\;\mu_i,\notag
\end{align}
where  $f_i^\delta$ is defined by \eqref{eq:barrObj}, and $g_i$ and $h_i$ are defined by \eqref{eq:BarrEConstr}, \eqref{eq:BarrIConstr}.
Define the Lagrangian to \eqref{eq:barrProbRewrite},
\begin{align*}
	L_i^\delta(x_i,s_i,\gamma_i,\mu_i;y) \doteq f_i^\delta (x_i,s_i;y) + \gamma_i^\top g_i(x_i;y) + \mu_i^\top(h_i(x_i;y) + s_i).
\end{align*}
Assume that \eqref{eq:barrProbRewrite} is feasible for a given $y$ and that the  regularity conditions from \cite[Ass. 1-4]{DeMiguel2008} hold. Then, the KKT conditions to \eqref{eq:barrProbRewrite} form an implicit function in form of $F_i^\delta(x_i^\star,s_i^\star,\gamma_i^\star,\mu_i^\star;y)=0$, where the superscript $(\cdot)^\star$ indicates a KKT stationary point.
Thus, by the implicit function theorem, there  exist a  neighborhood around $y$ for which there exists functions $p_i^\star(y)\doteq(x_i^\star(y),s_i^\star(y),\gamma_i^\star(y),\mu_i^\star(y))$ such that $F_i^\delta(p_i^\star(y);y)=0$. Hence, we can rewrite~\eqref{eq:barrProbRewrite} as $\Phi_i^\delta(y)= f_i(x_i^\star(y),s_i^\star(y);y)= L_i(p_i^\star(y);y)$ since $p_i^\star(y)$ is feasible.

Applying the total derivative and the chain rule yields
\begin{align*}
	\nabla_y 	\Phi_i^\delta(y) \hspace{-.1em}=\hspace{-.2em} \nabla_y L_i^\delta(p_i^\star(y);y) \hspace{-.2em}+\hspace{-.2em} \nabla_{p_i^{\star}}\hspace{-.1em} L_i^\delta(p_i^\star(y);y) \nabla_y p_i^{ \star}(y). \hspace{-.2em}
\end{align*}
By the KKT conditions, we have that $\nabla_{p_i^{\star}} L_i^\delta(p_i^\star(y);y)=0 $
and thus
\begin{align} \label{eq:nablaYPhi}
	\nabla_y 	\Phi_i^\delta(y) \hspace{-.1em}=\hspace{-.2em} \nabla_y L_i^\delta(p_i^\star(y);y).
\end{align}
Again by the total derivative, the Hessian can be computed by
\begin{align} \label{eq:Hess}
	\nabla_{yy}^2 	\Phi_i^\delta(y) = \nabla_{yy}^2 L_i(p_i^\star(y);y) + \nabla_{y p_i^\star}^2 L_i(p_i^\star(y);y) \,\nabla_y p_i^\star (y). 
\end{align}
It remains to derive an expression for $\nabla_y p_i^\star (y)$.
The KKT conditions of  \eqref{eq:barrProbRewrite} read
\begin{align*}
	F_i^\delta(x_i^\star,s_i^\star,\gamma_i^\star,\mu_i^\star;y) = 
	\begin{bmatrix}
		\nabla_{x_i} f_i(x_i^\star,y) + \nabla_{x_i} g_i(x_i^\star,y) \gamma_i^\star + \nabla_{x_i} h_i(x_i^\star,y) \mu_i^\star  \\
		-\delta S_i^{\star-1} \mathds 1 + \mu_i^\star  \\
		g_i(x_i^\star,y)  \\
		h_i(x_i^\star,y) + s_i^\star
	\end{bmatrix}
	\overset{!}{=}0, \notag
\end{align*}
where  $S_i^\star = \operatorname{diag}(s_i^\star)$. By the total differential and the chain rule we have $\nabla_y F_i^\delta(p_i^\star (y),y)  + \nabla_{p_i^\star} F_i^\delta(p_i^\star (y),y)  \nabla_y p_i^\star(y) =0 $. Hence,  we can compute the Jacobian $	\nabla_y p_i^\star(y)$ by solving the system of linear equations
\begin{align} \label{eq:dpStarDy}
	\left (\nabla_{p_i^\star} F_i^\delta(p_i^\star (y),y) \right ) \nabla_y p_i^\star(y)=- \nabla_y F_i^\delta(p_i^\star (y),y).
\end{align}
 Observe that \eqref{eq:dpStarDy} is a system of linear equations with multiple right-hand sides.
In summary, we can compute $\nabla_{yy}^2 	\Phi_i^\delta(y)$ locally for each $i \in \mathcal S$ by combining \eqref{eq:Hess} and \eqref{eq:dpStarDy}. The corresponding formulas for the gradient and the Hessian of  $\Phi_i^{\delta,\rho}$ and $\Phi_i^{\delta,\bar \lambda}$ from \eqref{eq:subProbBarrQP2} and \eqref{eq:subProbBarrQP3}, i.e. of the AL relaxation and the $\ell1$ relaxation \eqref{eq:barrProbRewrite}  are given in Appendix~\ref{sec:QPsens}.

 \section{A Method for Solving the Master Problem}

 An important question is how to solve the master problem~\eqref{eq:valFunReform} for different variants of $\phi_i$. In general, this can be done by any sensitivity-based NLP solver. We proceed by showing how to obtain a simple globalized version of \autoref{alg:basPrimDec} based  on a line-search  scheme; here, the idea is to show global convergence for the relaxed problem~\eqref{eq:valFunReform} with $\phi_i \in \{\Phi^{\delta,\rho}_i,\Phi_i^{\delta,\bar \lambda}\}$ for {fixed} penalty and barrier parameters. This leads to converge of a solution to the original problem~\eqref{eq:sepQP} by standard results from penalty and barrier methods \cite[Thms. 17.1, 17.6]{Nocedal2006}.

 Define the objective of \eqref{eq:valFunReform}, $\psi(y)\doteq \sum_{i \in \mathcal S} \phi_i(y)$,  as a {global} merit function, where $\phi_i \in \{\Phi^{\delta,\rho}_i,\Phi_i^{\delta,\bar \lambda}\}$. 
The basic idea is to employ a Sequential Quadratic Programming (SQP) scheme, where we ensure a sufficient decrease in $\psi$ at each step via the Armijo condition. The overall algorithm is summarized in \autoref{alg:maProbSolve}. Similar to the general primal decomposition scheme from \autoref{alg:basPrimDec}, the master problem solver evaluates the sensitivities $(\nabla_y \phi_i, \nabla_{yy}^2 \phi_i)$  in step (i), in order to construct a quadratic approximation of \eqref{eq:valFunReform} in step (ii). Solving this approximation yields a search direction $\Delta y$.
The stepsize $\alpha$ is updated with
a backtracking line-search with the Armijo condition as termination criterion.

 \subsection*{Global Convergence}

 We now establish global convergence\footnote{We use the  definition of global convergence in the context of NLPs, i.e. the convergence to a KKT point from an arbitrary initialization \cite[Chap. 3]{Nocedal2006}. However, since problems \eqref{eq:sepQP} and \eqref{eq:valFunReform} are strongly convex and we assume regularity in Assumption~\autoref{ass:stdAss}, every KKT point  is also a global  minimizer.} of \autoref{alg:maProbSolve} to a minimizer of the relaxed problems \eqref{eq:valFunReform} for  $\phi_i \in \{\Phi^{\delta,\rho}_i,\Phi_i^{\delta,\bar \lambda}\}$.
Assume that the following regularity assumptions hold at the optimal solution at  $p^\star\doteq[p^\star_i(y^\star)]_{i\in\mathcal S}$.
  \begin{ass}[Subproblem regularity] \hspace{-.28cm} \label{ass:stdAss}Assume that $\forall i \in \mathcal S$:
 	\begin{enumerate}
 		\item [a)]
 		$
 		\begin{bmatrix}
 			\Delta x_i^\star \\ \Delta  y^\star
 		\end{bmatrix}^	{\hspace{-.3em}\top}
 		\hspace{-.4em}
 		\begin{bmatrix}
 			H_i^{xx}\hspace{-.3em} & H_i^{xy} \\
 			H_i^{xy\top}\hspace{-.3em} & H_i^{yy}
 		\end{bmatrix}
 		\hspace{-.4em}
 		\begin{bmatrix}
 		\Delta 	x_i^\star \\ \Delta  y^\star
 		\end{bmatrix}>0,$ for all $ 		\begin{bmatrix}
 			\Delta 	x_i^\star \\\Delta  y^\star
 		\end{bmatrix} \neq 0$ with
 	 $
 	\begin{bmatrix}
 		A_i^x & \hspace{-.4em}A_i^y
 	\end{bmatrix}
 	\hspace{-.3em}
 	\begin{bmatrix}
 		\Delta 	x_i^\star \\\Delta  y^\star
 	\end{bmatrix}
 	\hspace{-.2em}=\hspace{-.1em}0
 	$;
 		\item [b)]
 		$
 		\begin{bmatrix}
 			A_i^x & A_i^y \\
 			B_i^x & B_i^y
 		\end{bmatrix}
 		$ has full row rank;
 		\item [c)] $[\mu_i^\star]_j \hspace{-.1em}+\hspace{-.1em}\left [[B_i^x\;\; B_i^y][x_i^{\star\top}\; y^{\star\top} ]^\top  \right  ]_j\hspace{-.3em} \neq 0$, $\forall j=1,\dots,\operatorname{nr}(B_i^x)$.
 	\end{enumerate}
 \end{ass}
  \begin{ass}[Master problem  regularity]\label{ass:Mreg} Assume that
		$
		\begin{bmatrix}
		A^{y\top} & B^{y\top}
		\end{bmatrix}^\top$ has full row rank.
\end{ass}
 Line-search methods require that the search direction $\Delta y$ is a descent direction, i.e. $\Delta y \left (\sum_{i \in \mathcal{S}} \nabla_i \phi_i (y) \right )< 0$.
  This can be ensured by showing that  $\sum_{i \in \mathcal{S}} \nabla_{yy}^2 \phi_i \succ 0$ for all variants of $\phi_i$, which we do with the next lemma.
 \begin{lem}[Positive definite Hessians] \label{lem:posDefHess}
 	Let Assumption~\autoref{ass:stdAss} hold and assume that $(s_i,\mu_i,v_i,w_i)>0$.
 	Then, a), the Hessian $	\nabla_{yy}^2 	\Phi_i^{\delta,\rho}$ is  positive definite for all $\rho >0$. Moreover, b), $	\nabla_{yy}^2 	\Phi_i^{\delta,\bar \lambda}$ is positive definite if  $\bar \lambda$ is larger than all multipliers associated to the elastic constraints $\bar \lambda > \max_j{\left |[\chi_i^\star]_j \right |}$.
 \end{lem}
The proof of Lemma~\autoref{lem:posDefHess} is given in Appendix~\ref{sec:ProofPosDef}.
Now we are able to show global convergence of  \autoref{alg:maProbSolve} to the solution of problem \eqref{eq:valFunReform} with  $\phi_i \in \{\Phi^{\delta,\rho}_i,\Phi_i^{\delta,\bar \lambda}\}$ for fixed penalty and barrier parameters.
 \begin{thm}[{Convergence of \autoref{alg:maProbSolve} for fixed $\delta, \rho,\bar \lambda$}] \label{thm:conv}
 	Consider \autoref{alg:maProbSolve} with either fixed $(\delta, \rho,\lambda)$ if $\phi_i = \Phi_i^{\delta,\rho}$ or with fixed $(\delta, \bar \lambda)$ if $\phi_i = \Phi_i^{\delta,\bar \lambda}$.
 	Let Assumptions~\ref{ass:stdAss}, \ref{ass:Mreg} and the conditions of Lemma~\autoref{lem:posDefHess} hold.
 	Then, the iterates generated by \autoref{alg:maProbSolve} converge to the global minimizer of problem~\eqref{eq:valFunReform} with $\phi_i \in \{\Phi^{\delta,\rho}_i,\Phi_i^{\delta,\bar \lambda}\}$.
 \end{thm}
  \begin{proof}
  	The unconstrained minimizer to \eqref{eq:coordPob},
    $\Delta y = - \sum_{i\in \mathcal S} \nabla_{yy}^2\phi_i ^{-1}(y)\sum_{i\in \mathcal S}  \nabla_{y} \phi_i(y)$, is a descent direction for the merit function $\psi(\cdot)$ since $\sum_{i\in \mathcal S}  \nabla_{y} \phi_i(y)^\top \Delta y < 0$ by the positive definiteness of the Hessians from Lemma~\autoref{lem:posDefHess}.
  	 Observe that $\Delta y =0$ is feasible for~\eqref{eq:coordPob} by feasibility of $y$.
  	 This shows that either $\Delta y=0 $, or $\Delta y$ is a descent direction.
Hence, by \cite[Lem 3.1]{Nocedal2006}, there  exists an $\alpha \in (0,1]$ such that $\psi(y) - \psi(y + \alpha \Delta y) \geq - \sigma \alpha\nabla_y \psi(y)^\top \Delta y $ is satisfied and thus the inner while loop is well defined. Moreover, by the convergence of line-search methods \cite[Prop 1.2.1]{Bertsekas1999}, \autoref{alg:maProbSolve} will either converge to a stationary point of $\psi$ returning $\Delta y =0$, or, alternatively, if the unconstrained search direction is blocked by the constraints in \eqref{eq:coordPob}, $\Delta y =0$ is returned since $\Delta y =0$ is feasible.

\begin{algorithm}[t]
	\SetAlgoLined
	\LinesNumbered
	\caption{A simple master problem solver. \label{alg:maProbSolve}}
	\BlankLine
	\textbf{Initialize} $y^0, \epsilon, \delta,  \zeta \in (0,1), \rho \text{ or }\bar \lambda, \phi_i \in \{\Phi^{\delta,\rho}_i,\Phi_i^{\delta,\bar \lambda}\}$. \\
	\While{ $\|\Delta y\| > 0$}{
		1) Evaluate $(\nabla_y \phi_i, \nabla_{yy}^2 \phi_i)$ locally for all  $i \in \mathcal S$ for given $(\delta,\rho)$ or $(\delta, \bar \lambda)$. \\
		2) Solve the coordination problem
		\begin{align} \label{eq:coordPob}
			\min_{\Delta y} &\sum_{i \in \mathcal S} \frac{1}{2} \Delta y\nabla_{yy}^2 \phi_i (y)\Delta y^\top \hspace{-.4em}+ \nabla_y \phi_i(y)^\top \hspace{-.2em}\Delta y\\
			\text{ s.t.} \,	&A^y (y \hspace{-.1em}+\hspace{-.2em} \Delta y) -b^y \hspace{-.2em}= \hspace{-.2em}0,	B^y (y\hspace{-.1em} + \hspace{-.2em}\Delta y) -b^y\hspace{-.2em} \leq\hspace{-.2em}  0. \notag
		\end{align}

		3) Line search: $\alpha = 1$;

		\While{\textnormal{$\psi(y) - \psi(y + \alpha \Delta y) \geq \hspace{-.2em} - \sigma \alpha\nabla_y \psi(y)^\top \Delta y $}}{
			$\alpha \leftarrow \zeta \alpha$}
		4) Update $y \leftarrow y + \alpha \Delta y$. \\
		5) Decrease $\delta$, increase $(\rho,\bar \delta).$\\
	}
	\textbf{Return} $y$, $\{x_i\}_{i\in \mathcal S}$.
	\BlankLine
\end{algorithm}

We now show that $y$ is optimal for \eqref{eq:valFunReform} if $\Delta y =0$.
The KKT conditions associated to \eqref{eq:coordPob} read
\begin{equation*}
\left\{
\begin{aligned}
	\sum_{i \in \mathcal S} \nabla_{yy}^2 \phi_i (y) \Delta y + \nabla_y \phi_i(y)^\top + A^{y\top}\gamma_y + B^{y\top}\mu_y &=0 \\
 A^y (y + \Delta y) -b^y = 0, \quad
 	B^y (y + \Delta y) -b^y &\leq  0,\\
(B^y (y + \Delta y) -b^y)^\top \mu =0,\quad  \mu&\geq 0,
\end{aligned}
\right.
\end{equation*}
which are precisely the KKT conditions for \eqref{eq:valFunReform} if $\Delta y =0$. Since \eqref{eq:valFunReform} is convex, and Assumptions~\ref{ass:stdAss} and \ref{ass:Mreg} hold, the assertion follows.
 \end{proof}

Combining Theorem~\autoref{thm:conv} with the convergence results for the $\ell1$-penalty method \cite[Them 17.3]{Nocedal2006} or the augmented Lagrangian method \cite[Prop 2.7]{Bertsekas1982} implies convergence of \autoref{alg:maProbSolve} to the minimizer of the original problem \eqref{eq:sepQP} for sufficiently large penalties.

\begin{remark}[Satisfying the assumptions of Theorem~\autoref{thm:conv}]
	Observe that the assumptions for Theorem~\autoref{thm:conv} are standard regularity assumptions from nonlinear programming \cite[Ass. 1-4]{DeMiguel2008}.
	Moreover, $(s_i,\mu_i,v_i,w_i)>0$ is always ensured when using interior-point solvers for solving \eqref{eq:ValueFunBarr2} and \eqref{eq:subProbBarrQP2} even in the case of early termination.
\end{remark}

\begin{remark}[Convergence rate] \label{rem:convRate}
Since phase~1 of \autoref{alg:ALPrimDec} is a combination of a quadratic penalty and an interior point method, one can achieve up to a superlinear local convergence rate  when decreasing/increasing $\delta/\rho$ fast enough since the only limiting factor is the Newton step \cite[Chap 17.1, Chap 19.8]{Nocedal2006}.
A too fast increase/decrease, however, leads to numerical difficulties in Newton step computations.
In phase~2, the algorithm switches to a augmented Lagrangian method, where one can expect  linear  convergence to the barrier problem~\eqref{eq:ValueFunBarr}, cf.  \cite[Prop 2.7]{Bertsekas1982}.
\end{remark}

\section{Implementation Aspects} \label{sec:Implementation}

The evaluation of the sensitivities of $\phi_i \in \{\Phi^{\delta,\rho}_i,\Phi_i^{\delta,\bar \lambda}\}$ requires solving local optimization problems \eqref{eq:subProbBarrQP2} or \eqref{eq:subProbBarrQP3} for fixed $\delta,\rho,\bar \lambda$.
This can be done using  specialized and optimized interior-point solvers, if they allow termination once a certain barrier $\delta$ is reached.
Moreover, interior-point solvers factorize the KKT matrices $\nabla_{p_i}F_i^\delta$ (cf. \eqref{eq:KKTsysAL}, \eqref{eq:KKTsl12}) at each inner iteration and these factorizations can be re-used for Hessian computation via \eqref{eq:dpStarDy}.
Here we provide two variants: our own interior-point QP solver based on standard techniques for stepsize selection and barrier parameter decrease \cite[Chap. 16.6]{Nocedal2006} and the option to use third-party solvers such as \texttt{Ipopt}~\cite{Wachter2006}.

In early iterations, it is typically not necessary to solve the local problems to a high precision, since the barrier parameter $\delta$ is still large and the penalty parameters $(\rho,\bar{\lambda})$ are still small.
Hence, we solve the subproblems to an accuracy measured in the violation of the optimality conditions $\|F_i^\delta(p_i^k,y^k)\|_\infty$ and terminate if  $\|F_i^\delta(p_i^k,y^k)\|_\infty <\min(\delta,1/{\rho})$ or $\|F_i^\delta(p_i^k,y^k)\|_\infty <\min(\delta,1/{\bar \lambda})$. This is inspired by the termination of inexact interior-point methods~\cite{Byrd1998}. Warm-starting the local solves with the solution of the previous iteration reduces computation time significantly.

\subsubsection*{Numerical Linear Algebra}
Efficient numerical linear algebra is crucial for performance.
The most expensive steps in terms of memory and CPU time are the factorization of the local KKT matrices and the backsolves in \eqref{eq:dpStarDy}.
Since we only consider QPs, large parts of the KKT matrices in \eqref{eq:dpStarDy} are constant over the iterations, which can be exploited for pre-computation. Here, we make heavy use of the Schur complement. Details on how to achieve this for the AL formulation are given in Appendix~\ref{sec:QPsens}.

\subsubsection*{Updating Penalties}
We use simple update rules for the penalty parameters. They are
\begin{align*}
	\delta^{k+1}= 0.2\delta^k, \quad \bar \lambda^{k+1} =  2 \bar\lambda^k, \quad \rho^{k+1} = 3\rho^k,
\end{align*}
with initial values  $\delta^0= 0.1, \rho^0 = 10^3$, and $\bar \lambda^0 = 100$.
After a fixed amount of $8$ iterations, phase 2 starts and the above values stay constant.
The parameters in the update rules are found via tuning  in typical ranges for interior point methods, augmented Lagrangian and $\ell$1-penalty methods \cite[Chapters 17.1, 17.2, Eq. 19.18]{Nocedal2006}.  They  terminate with  $(\delta^\infty,\bar \lambda^\infty,\rho^\infty) = (10^{-8},\;2.6\cdot 10^4, \;6.6\cdot 10^6)$. Moreover,  \texttt{IPOPT} terminates with penalty $\delta^\infty = 10^{-8}\dots10^{-9}$ for the examples considered.\footnote{Note that \texttt{IPOPT} uses $\mu$ as symbol for the penatly parameter.}

\section{Numerical Case Studies} \label{sec:numRes}
We consider an optimal control problem for  a city district with a scalable number of commercial buildings connected via a electricity grid with limited capacity. The building data is from \cite{Rawlings2018}. We neglect the waterside HVAC system and assume that  the buildings are equipped with  heat pumps with a constant coefficient of performance.

\begin{figure}
	\centering
	\includegraphics[width=0.5\linewidth]{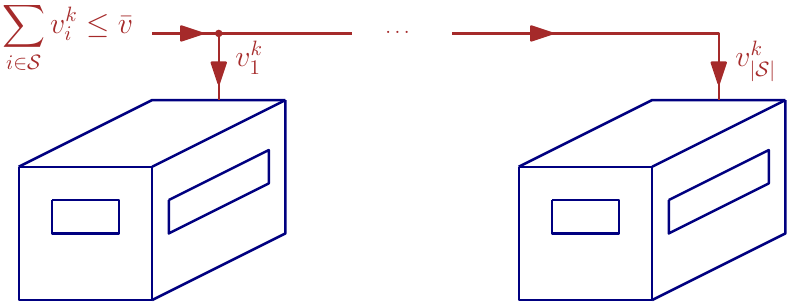}
	\caption{Buildings connected via network with limited capacity.}
	\label{fig:building}
\end{figure}

\subsection{District HVAC}
The evolution of the temperature of the $m$th zone in the $i$th building reads
\begin{align} \label{eq:roomDyn}
	C_{m,i} (T^{k+1}_{m,i} - T^k_{m,i}) =& - \hspace{-.2em}H_{m,i} (T_{m,i}^k \hspace{-.2em}- T_a^k) - \hspace{-.5em}\\
& \sum_{n \in \mathcal{Z}_i } \hspace{-.2em}G_{mn,i} (T_{m,i}^k \hspace{-.2em}- T_{n,i}^k) - \dot Q^{ck}_{m,i} \hspace{-.2em}+ \dot Q^{dk}_{m,i}, \notag
\end{align}
where at time step $k$, $T^k_{m,i}$ is the temperature of zone $m$ and $T_a^k$ the ambient temperature, $C_{i,m}$ is the thermal capacity, $H_{m,i}$ and $G_{mn,i}$ are heat transfer coefficients with the ambient and between two zones. Moreover,  $(\dot Q^{ck}_{m,i}, \dot Q^{dk}_{m,i}) $ are the controllable/uncontrollable heat influxes from the heat pump and from sources of disturbance such as solar irradiation and occupancy.
Eq.~\eqref{eq:roomDyn} can be written in compact form as
\begin{align*}
	C_i (T^{k+1}_i - T^k_i) = -H_i (T^k_i - T_a^k) - G_i T^k_i - \dot Q_i^{ck} + \dot Q^{dk}_i,
\end{align*}
where $T_i^\top  \doteq[T_{1,i}, \dots T_{|\mathcal{Z}_i|,i}], \dot Q_i^{ck^\top }  \doteq[\dot Q^{ck}_{1,i}, \dots \dot Q^{ck}_{n_{z},i}]$ and $\dot Q_i^{dk^\top }  \doteq[\dot Q^{dk}_{1,i}, \dots \dot Q^{dk}_{n_z,i}]$.
This yields a state-space model
\begin{align}
 T^{k+1}_i \doteq 	z^{k+1}_i&= \left (I -C^{-1}_i(H_i+G_i) \right )T^k_i - C_i^{-1}\dot Q^{ck}_i + C_i^{-1}[H_i\quad I] [T_a^{k\top} \quad \dot Q_i^{dk\top} ]^\top \label{eq:stateSp} \\
	&\doteq A_ix^k + B_iu^k_i +  E_id^k_i.\notag
\end{align}
Stacking the above over $N$  time steps yields
\begin{align*}
	\bar z_i &=  
	\begin{bmatrix}
		0 & 0 \\
		I_{N-1} \otimes A_i & 0
	\end{bmatrix}
	\bar z_i
	 +  
	\begin{bmatrix}
		0 \\
		I_{N-1} \otimes\hspace{-.2em} B_i
	\end{bmatrix}
\bar u_i
	 +  
	\begin{bmatrix}
		0 \\
		I_{N-1}\otimes E_i
	\end{bmatrix}
	\bar d_i
	 +  
\bar z_i^f
	\\
	\hspace{-.2em}&\doteq \bar A_i \bar z_i + \bar B_i \bar u_i + \bar E_i \bar d_i + \bar z_i^f,
\end{align*}
where $\bar z_i^\top \doteq [z_i^{0\top}, \dots, z_i^{N\top }]$, $\bar u_i^\top \doteq [u_i^{1\top}, \dots, u_i^{N\top }]$, $\bar d_i^\top \doteq [d_i^{1\top}, \dots, d_i^{N\top }]$, $\bar  z_i^{f\top} \doteq [z_i^f,0^\top , \dots, 0^\top]$ and  $z_i^f$ are the initial temperatures.
Define the total energy consumption of building $i \in \mathcal S$ at time step $k$ by $v_i^k = \mathds 1^\top u_i^k$, and $\bar v_i^\top \doteq [v_i^{0\top}, \dots, v_i^{N\top }]$.
Then, the above is equivalent to
\begin{align} \label{eq:dynCstr}
	\begin{bmatrix}
		(	I- \bar A_i) & -\bar B_i & 0
	\end{bmatrix}
	\begin{bmatrix}
		\bar z_i^\top &
		\bar u_i^\top  &
		\bar v_i^\top
	\end{bmatrix}^\top
	= \bar E_i \bar d_i + \bar z_i^f.
\end{align}
\Fp{The grid coupling between all subsystems $i \in \mathcal S$ induces an upper-bounded energy supply
writing as a global constraint:
$\mathds 1^\top v_i^k \leq \bar v$ for all times $k$.}
Moreover, we have local comfort constraints $\underline T \mathds 1 \leq z_i^k \leq  \mathds 1 \bar T$.

\subsection*{Optimal Control Problem}
Assume that the  goal of each building is to minimize its cost of energy consumption over $N$ time steps respecting all constraints.
The cost function to buy the power form the utility is given by  $f_i^k(u_i^{k}) \doteq 0.5\,c^k (u_i^{k})^2 + g^k u_i^k$, where  $g^k$ is a linear and  $c^k > 0$ is a (small) quadratic cost coefficient.
This yields a discrete-time Optimal Control  Problem~(OCP) for building~$i$,
\begin{align*}
	\phi_i(v_i) &\doteq \hspace{-.5em}\min_{\bar z_i, \bar u_i, \bar v_i}
	\frac{1}{2} \hspace{-.3em}
	\begin{bmatrix}
		\bar z_i \\ \bar u_i \\ \bar  v_i
	\end{bmatrix}
^{\hspace{-.5em}\top}
\hspace{-.5em}
	\begin{bmatrix}
		0 & 0 & 0 \\
		0 & c I  & 0\\
		0 & 0 &0
	\end{bmatrix}
\hspace{-.5em}
	\begin{bmatrix}
		\bar x_i \\ \bar u_i \\\bar v_i
	\end{bmatrix}
\hspace{-.3em}+\hspace{-.3em}
	\begin{bmatrix}
	0 \\  \mathds 1^\top \otimes g \\ 0
\end{bmatrix}^{\hspace{-.5em}\top} \hspace{-.5em}
\begin{bmatrix}
	\bar z_i \\ \bar u_i \\ \bar v_i
\end{bmatrix}
	\\
	\text{s.t. } &\; \eqref{eq:dynCstr},
	\begin{bmatrix}
		I & 0 & 0  \\
		-I & 0 & 0
	\end{bmatrix}
	\begin{bmatrix}
		\bar z_i \\ \bar u_i \\v_i
	\end{bmatrix}
	\leq
	\begin{bmatrix}
		\bar T \mathds 1 \\ - \underline T \mathds 1
	\end{bmatrix}, \;
	\bar v_i = I \otimes \mathds 1^\top  \bar u_i.
\end{align*}
The overall OCP---including global grid constraints---reads
\begin{align} \label{eq:OCP}
	\min_{v_1,\dots,v_{|\mathcal S|}} \sum_{i \in \mathcal S} \phi_i(v_i), \text{ s.t } (\mathds 1^\top \otimes I) \,[v_1^\top , \dots, v_{|\mathcal S|}^\top]^\top  \leq \bar v \, \mathds 1.
\end{align}
To obtain a problem in form of \eqref{eq:sepQP}, define  $x_i^\top = [\bar x_i^\top, \bar u_i^\top, \bar v_i^\top], i \in \mathcal S$, $y^\top = [v_1^\top ,\dots,v_{|\mathcal S|}^\top]$, $ H_i^{x}= \operatorname{blkdiag}(0, c I,0),h_i^{x} = [0, \quad  \mathds1 ^\top \otimes g,0 ]^\top, h_i^{y} = 0, H_i^{xy}= 0, B_i^y = 0, d_i^\top  = [\bar T \mathds 1^\top, \; -\bar T \mathds 1^\top  ]$,
\begin{align*}
	A_i^x = 
	\begin{bmatrix}
		(	I- \bar A_i)& -\bar B_i & 0  \\
		0 & I \otimes \mathds 1^\top & 0
	\end{bmatrix}, \quad 
	A_i^y = 
	\begin{bmatrix}
		0 \\
		-e_i^\top  \otimes I
	\end{bmatrix}, \quad 
	B_i^x  =
	\begin{bmatrix}
		I & 0&0   \\
		-I& 0 &0
	\end{bmatrix},
\end{align*}
where $e_i$ is the $i$th unit vector, $b_i^\top  =
\begin{bmatrix}
	(\bar E_i \bar d_i + \bar z^f_i)^\top  & 0
\end{bmatrix}$, $A^y =0 $,   $b^y=0$, $B^y = \mathds 1 ^\top \otimes I $, and  $d^y=\bar v \mathds 1$.

\subsection{Optimal Power Flow} \label{sec:OPFform}
Optimal Power Flow (OPF) aims at minimizing the cost of power generation in power systems while satisfying all grid and generator constraints.
A standard OPF formulation   reads
\begin{subequations}  \label{eq:OPF}
\begin{align}
	&\min_{g,\theta, f}\; \frac{1}{2} g^\top Hg + g^\top h \\
	\text{subject to }  \quad & C^gg-d = -B \theta, \quad  f = -B^b \theta, \label{eq:DC_PFEQ} \\
	& 0 \leq g \leq \bar g, \quad -\bar f \leq f \leq \bar f, \label{eq:genFlowConstr}
\end{align}
\end{subequations}
where $g \in \mathbb R^{n_g}$ is the  active power generation of generators, the diagonal matrix $H$ and vector $h$ are composed of generator-specific cost coefficients, $d \in \mathbb{R}^{n_b}$ are  active power demands and $\theta \in \mathbb R^{n_b}$ are voltage angles at each bus.
In \eqref{eq:DC_PFEQ}, $B \in \mathbb{R}^{n_b\times n_b}$ is the bus susceptance matrix and $B^b\in \mathbb R^{n_B\times  n_G}$ is the branch susceptance matrix, which map the voltage angles to power injections and power flows over transmission lines $f \in \mathbb{R}^{n_l}$ respectively, cf. \cite{Molzahn2019}.
The matrix $C^g \in \mathbb{R}^{n_b\times n_g}$ maps generator injections to connecting buses.
The constraint \eqref{eq:genFlowConstr} expresses generation and line flow limits.

Power grids are typically structured in hierarchy levels reaching from extra-high voltage to low-voltage grids.
As a numerical test case, we consider the \texttt{IEEE 300-bus}  test system to which we connect a varying amount of \texttt{118-bus} sub-grids (with data from the \texttt{MATPOWER} database \cite{Zimmerman2011}). We add a small regularization term of $10^{-6}$ on the main diagonal of each $H_{yy}$ to make the problem strongly convex in order to meet the conditions of Assumption~\ref{ass:stdAss}.

To obtain a problem in form of \eqref{eq:sepQP}, we introduce decision variables for the master grid $y = [g_0 \;\;\theta_0\;\; f_0\;\; \bar y_0]$, where $\bar y_0$ is an auxiliary variable corresponding to the active power at interconnecting nodes with the lower-level network.
For the $i$th lower-level subproblem we get
\begin{subequations} \label{eq:distOPF}
	\begin{align*}
		\phi_i(y)\hspace{-.1em} \doteq \hspace{-.1em}	\min_{g_i,\theta_i, f_i }
		\frac{1}{2}
		\hspace{-.1em}
		\begin{bmatrix}
		 g_i \\	\theta_i \\ f_i \\ y
		\end{bmatrix}^\top
		\hspace{-.9em}
		\begin{bmatrix}
			H_i & 0 \\
			0 & 0
		\end{bmatrix}
		\hspace{-.3em}
		\begin{bmatrix}
 g_i\\ \theta_i  \\ f_i \\ y
		\end{bmatrix}
		\hspace{-.2em}
		+
		\hspace{-.2em}
		\begin{bmatrix}
			h_i \\ 0
		\end{bmatrix}^\top
		\hspace{-.3em}&
		\begin{bmatrix}
			 g_i \\	\theta_i \\ f_i \\ y
		\end{bmatrix} \\
	\text{s.t.}
		\begin{bmatrix}
C_{i}^g & B_i & 0 & C_i^y \\
0 & B_i^b & I & 0
		\end{bmatrix}
		\begin{bmatrix}
				 g_i &	\theta_i & f_i & y
		\end{bmatrix}^\top  &=
	\begin{bmatrix}
	d_i \\ 0
	\end{bmatrix}, \label{eq:distPFEQ}\\
		\begin{bmatrix}
			I & 0 & 0 & 0 \\
			-I & 0 & 0 & 0 \\
			0 & 0 & I & 0 \\
			0 & 0 & -I & 0 \\
		\end{bmatrix}
		\begin{bmatrix}
			 g_i &	\theta_i & f_i & y
		\end{bmatrix}^\top &\leq
	\begin{bmatrix}
		\bar g_i \\ 0 \\ \bar f_i \\ \bar f_i
	\end{bmatrix}.
	\end{align*}
\end{subequations}
Here, $C_i^y$ are a selection matrices, which couple power demand/generation at interconnecting nodes between subsystems.
The master problem then reads
	\begin{align*}
\min_{y }
	 \sum_{i \in \mathcal{S}} \phi_i(y) \quad \text{subject to } \quad [\theta_0]_1 = 0,
		\\
		\begin{bmatrix}
			C_{0}^g & B_0 & 0 & C_0^y \\
			0 & B_0^b & I & 0
		\end{bmatrix}
		\begin{bmatrix}
			g_0 &	\theta_0 & f_0 & \bar y_0
		\end{bmatrix}^\top  &=
		\begin{bmatrix}
			d_0 \\ 0
		\end{bmatrix}, \\
		\begin{bmatrix}
			I & 0 & 0 & 0 \\
			-I & 0 & 0 & 0 \\
			0 & 0 & I & 0 \\
			0 & 0 & -I & 0 \\
		\end{bmatrix}
		\begin{bmatrix}
			g_0 &	\theta_0 & f_0 & \bar y_0
		\end{bmatrix}^\top &\leq
		\begin{bmatrix}
			\bar g_0 \\ 0 \\ \bar f_0 \\ \bar f_0
		\end{bmatrix},
	\end{align*}
where $[\theta_0]_1 = 0$ is a reference (slack) constraint in order to obtain an unique angle solutions $\{\theta_i^\star\}_{i \in \mathcal{S} \cup \{0\}}$, and  $C_0^y$ is a matrix mapping the coupling variables to  coupling buses.

\begin{table}[t]
	\centering
	\caption{Number of decision variables and constraints  for the HVAC and OPF problems.}
	\begin{tabular}{rrrrrrrr}
		\hline
		&$|\mathcal S|$	&  $n_x$ & $n_y$& $n_{e}$& $n_{i}$ & $n_{ey}$& $n_{iy}$\\
		\hline
		\parbox[t]{2mm}{\multirow{3}{*}{\rotatebox[origin=c]{90}{HVAC}}} 	&30	& 28,200 & 690 & 15,090 & 28,800 & 0 & 1,403  \\
		&180	& 169,200 &  4,140 & 90,540 & 172,800 & 0 & 8,303    \\
		&300	& 282,000 & 6,900 & 150,900 & 288,000 &  0 & 13,823  \\
		\hline
		\parbox[t]{2mm}{\multirow{2}{*}{\rotatebox[origin=c]{90}{OPF}}} 	&29 & 11,191 & 809 & 9,557 & 14,880 & 712 & 960 \\
		&	64& 23,756 & 844 &20,232&31,680 & 712 & 960 \\
		\hline
	\end{tabular} \label{tab:nDecConstr}
\end{table}

\subsection{Numerical Results} \label{sec:numRes2}
We benchmark our algorithms against ADMM (as one of the most popular algorithms for decomposition) and against \texttt{\texttt{Ipopt}} v3.14.4 (as one of the most prominent centralized NLP solvers).\footnote{The ADMM-based QP solver \texttt{OSQP} solver did not converge for the problems presented here.}
The particular variant of ADMM can be found in an extended version of this work \cite{Engelmann2022b}.
We rely on \texttt{OSQP} v0.6.2~\cite{Stellato2020} for solving subproblems and the coordination problem in ADMM.
In primal decomposition, we rely on our own interior-point solver for the subproblems and on \autoref{alg:maProbSolve} for coordination, where we solve \eqref{eq:coordPob} via \texttt{Ipopt}.\footnote{Note that the proposed framework is flexible with respect to  the interior point solvers used in the subproblems as long as one can access the corresponding sensitivity matrices. \label{fn:localIP}}
We perform all simulations on a shared-memory virtual machine with 30 cores and 100GiB memory.
The underlying hardware is exclusively used for the case studies.
All algorithms are parallelized via Julia multi-threading---thus all subproblems are solved on multiple cores in parallel.

We compare the numerical performance of all algorithms on OCP \eqref{eq:OCP} for $|\mathcal S|\in \{30,180, 300\}$ buildings, and on the OPF problem \eqref{eq:OPF} with  $|\mathcal S|\in \{29,64\}$ sub-grids.  \autoref{tab:nDecConstr} shows the corresponding number of local/global decision variables $n_x$/$n_y$,  the number of local equality/inequality constraints $n_{e}$/$n_{i}$, and the number of global equality/inequality constraints $n_{ey}$/$n_{iy}$.
We employ ADMM from \cite{Engelmann2022b} with penalty parameters $\rho\in \{10^0,10^1,10^2,10^3\}$.

\begin{figure}[t]
   \begin{subfigure}[b]{.45\textwidth}
	\centering
\includegraphics[width=1\linewidth]{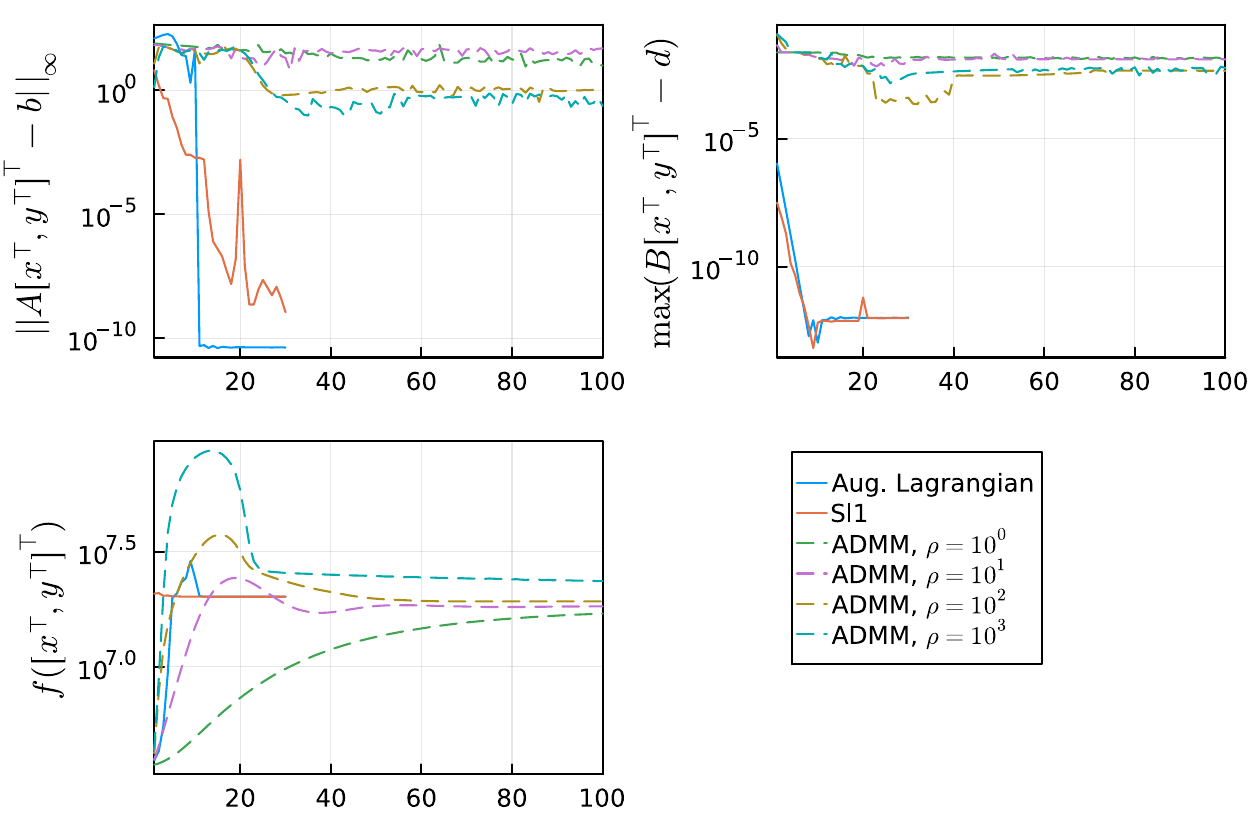}
\vspace{-.6cm}
\caption{$|\mathcal S|=30$ buildings.}
\label{fig:30subsyshvac}
\end{subfigure}
   \begin{subfigure}[b]{.45\textwidth}
	\centering
		\vspace{.1cm}
\includegraphics[width=1\linewidth]{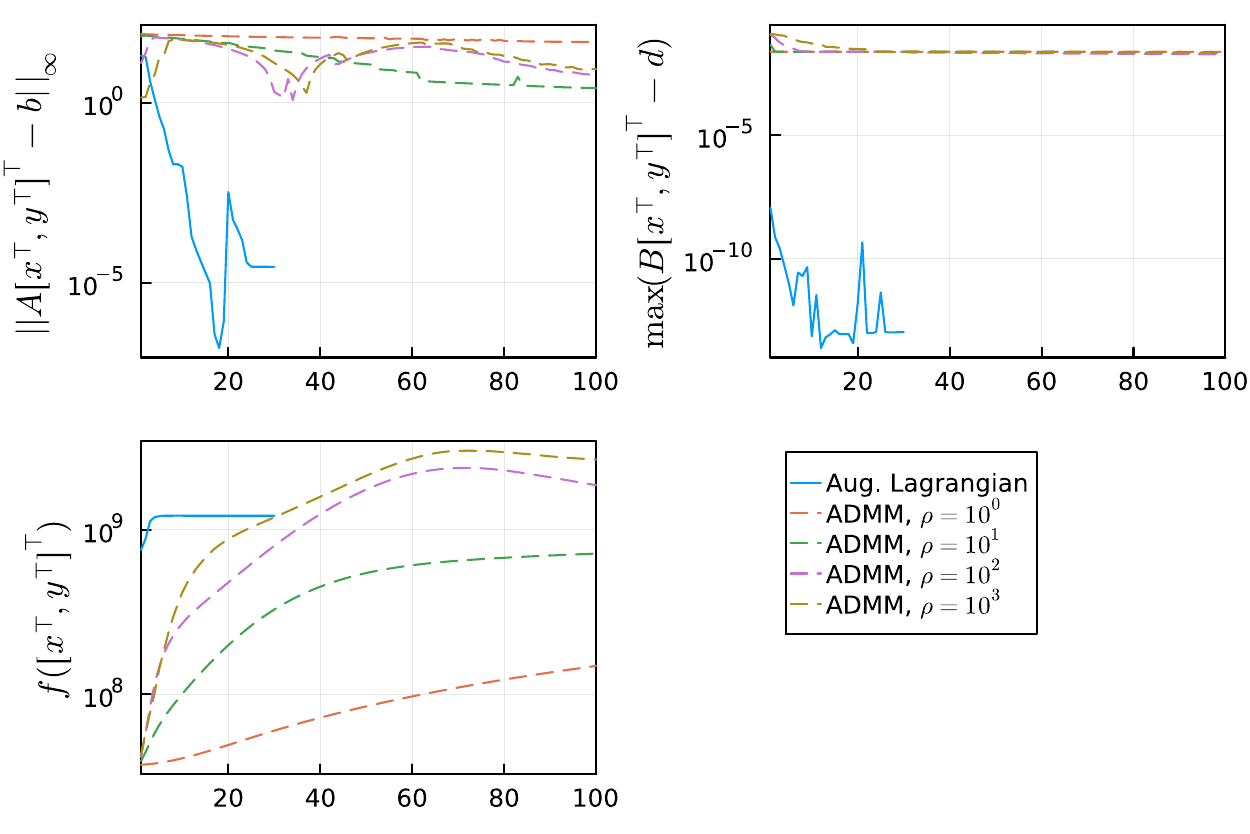}
\vspace{-.6cm}
\caption{$|\mathcal S|=300$ buildings.}
\label{fig:300subsyshvac}
\end{subfigure}
\caption{Convergence for three HVAC problems.}
\end{figure}

\begin{figure}[t]
	\begin{subfigure}[b]{.45\textwidth}
		\centering
		\includegraphics[width=1\linewidth]{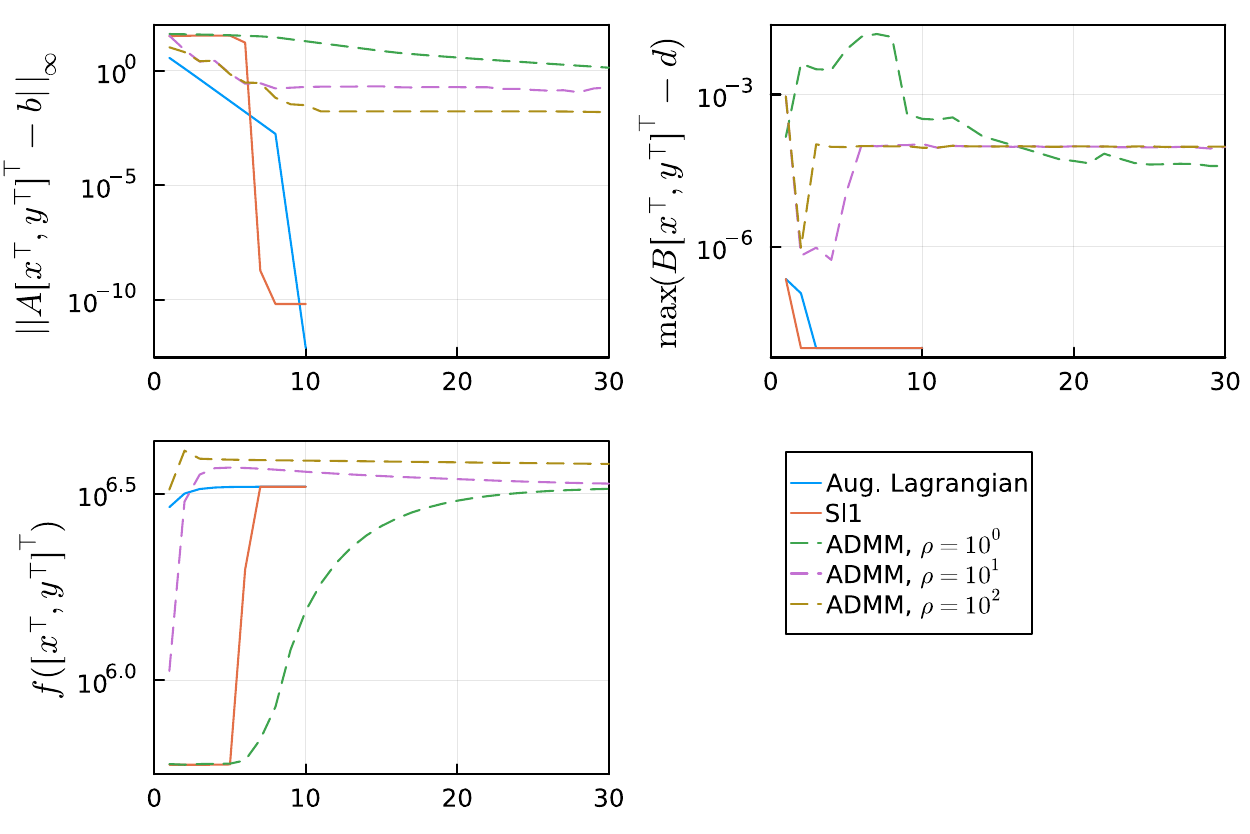}
		\vspace{-.6cm}
		\caption{$|\mathcal S|=29$ sub-grids.}
		\label{fig:29subsysOPF}
	\end{subfigure}
	\begin{subfigure}[b]{.45\textwidth}
		\centering
		\vspace{.1cm}
		\includegraphics[width=1\linewidth]{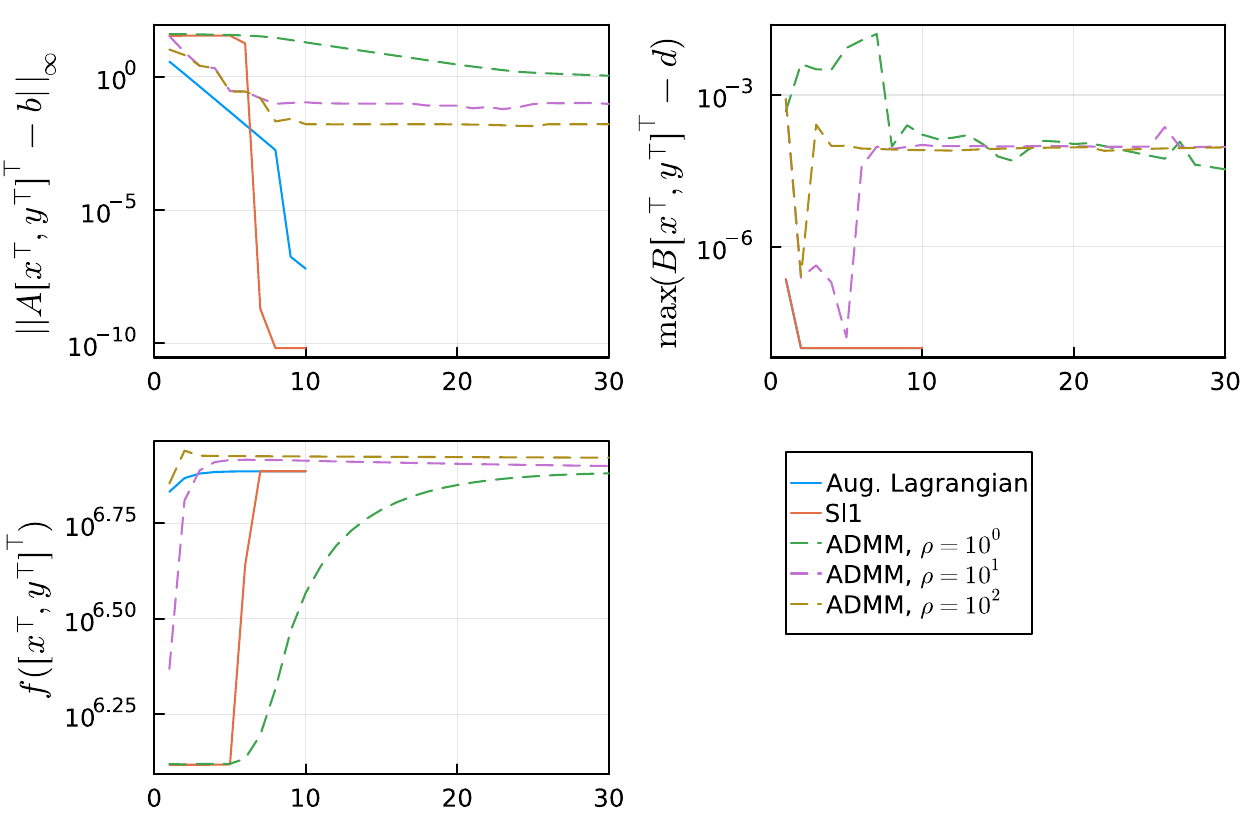}
		\vspace{-.6cm}
		\caption{$|\mathcal S|=64$ sub-grids.}
		\label{fig:64subsysOPF}
	\end{subfigure}
	\caption{Convergence for two OPF problems.} \label{fig:convOPF}
\end{figure}

\autoref{fig:30subsyshvac} illustrates the numerical performance of both primal decomposition variants and ADMM for the HVAC problems. \autoref{fig:convOPF} shows their performance for the OPF problems.
\autoref{fig:300subsyshvac} shows the AL formulation only, since the $\ell$1  formulation runs out of memory for this problem.
The constraint violations for the equality constraints \eqref{eq:localEq}, $
\left \|[
	A^x \;\; A^y
]
[
	x^\top \;\; y^\top
]^\top - b \right \|_\infty$,
for the inequality constraints \eqref{eq:localIneq},
$ \max([
B^x \;\; B^y
]
[
x^\top \;\; y^\top
]^\top - b)$, and the value of the cost function  $f([x^\top,y^\top]^\top)$ from \eqref{eq:sepCost} are displayed, where the x-axis shows the iteration count.
One can observe that the primal decomposition schemes achieve a high degree of feasibility in less than 10 iterations for all cases.
Moreover, the optimality gap $(f([x^\top,y^\top]^\top) - f([x^{\star \top},y^{\star \top}]^\top))/f([x^{\star \top},y^{\star \top}]^\top)$ is below $0.01 \%$ in less than 10 iterations for both primal decomposition variants and for all $|\mathcal S|$, where $[x^{\star \top},y^{\star \top}]^\top$ is computed via \texttt{Ipopt}.
For ADMM,  infeasibility  and the optimality gap stay large independently of the choice of~$\rho$.

\begin{remark}[Scaling of the $\ell$1 formulation]
	The reason for the poor scaling of the $\ell$1-formulation is two-fold: First, the relaxation \eqref{eq:subProbBarrQP3} introduces $2n_y$ additional slack variables and inequality constraints.
	Hence, the KKT system in the subproblems defined via \eqref{eq:KKTsl12}  has a larger size
  than the KKT system we get with the AL formulation \eqref{eq:KKTsysAL}.
	Moreover, the additional inequality constraints potentially lead to smaller stepsizes due to the fraction-to-boundary rule \cite[Eq 19.9]{Nocedal2006}. Hence, more iterations in the subproblems are required compared to the AL formulation.
\end{remark}

\section{Discussion of Algorithmic Properties} \label{sec:algProps}
Next, we discuss algorithm properties in view of the desirable properties from \autoref{sec:intro}.

\subsubsection*{Computation Times}
Comparing computation times between algorithms is difficult, since the numerical performance strongly depends on the implementation.
Nonetheless, we would like to provide some timing information to underline the potential of primal decomposition methods.
\autoref{tab:nDecConstrOPF} shows the computation times and the number of iterations for all algorithms.
\texttt{Ipopt} is terminated at optimality with default settings and the computation time for the primal decomposition schemes are evaluated once an optimality gap of $10^{-4}$ and a maximum constraint violation for equality/inequality constraints of $10^{-5}$ is reached.
ADMM is terminated after a maximum of 5,000 iterations.
One can observe that the AL formulation and \texttt{Ipopt}
have similar computation times independently of $|\mathcal{S}|$.
Although fast computation is not our primary focus, this indicates the potential of primal decomposition for large-scale optimization.

\begin{table}[t]
	\centering
	\caption{Timing and number of iterations for the HVAC and the OPF problem with $|\mathcal S|\in \{30,180, 300\}$ buildings and  $|\mathcal S|\in \{29,64\}$ sub-grids, 30 cores.}
	\begin{tabular}{rrcccccc}
		\hline
		&	&  &   &&\texttt{Ipopt}& ADMM  & ADMM  \\
		&	$|\mathcal S|$ & 	& AL& $l$1& par. LA&$\rho = 10$& $\rho = 100$ \\
		\hline
		\parbox[t]{2mm}{\multirow{3}{*}{\rotatebox[origin=c]{90}{HVAC}}}
		&300 &	t[s]& 431.5&  - & \textbf{386.7}  & 892.8$^*$ & 1,122.2$^*$ \\
		&180 &	& \textbf{195.7} & - &218.1 & 510.3$^*$& 1,322.2$^*$ \\
		&30& & \textbf{18.1} & 270.0&25.5 & 72.8$^*$ &  85.53$^*$\\
		\hline
		\parbox[t]{2mm}{\multirow{2}{*}{\rotatebox[origin=c]{90}{OPF}}}
		&64 &	t[s]& \textbf{2.64}&  283.89 & 4.64 & 70.52$^*$ & 111.29$^*$ \\
		&29 &	& \textbf{1.83} & 95.76 &2.01 & 13.54 & 61.58$^*$ \\
		\hline
		\parbox[t]{2mm}{\multirow{3}{*}{\rotatebox[origin=c]{90}{HVAC}}}
		&300&	iter. &   \textbf{13} & -& 145 &5,000$^*$&5,000$^*$\\
		&180&	 &  \textbf{12} & -& 141 & 5,000$^*$&5,000$^*$ \\
		&30 &		& 13&\textbf{12}& 104 &5,000$^*$&5,000$^*$\\
		\hline
		\parbox[t]{2mm}{\multirow{2}{*}{\rotatebox[origin=c]{90}{OPF}}}
		&64&	iter. &   \textbf{9} & \textbf{9}& 16 & 3,199$^*$ & 5,000$^*$ \\
		&29&	 &  9 & \textbf{7}& 15 & 1,133 &5,000$^*$ \\
		\hline
		&&  term. & \multicolumn{2}{c}{rel. opt. $10^{-4}$ } & optimal &  \multicolumn{2}{c}{rel. opt. $10^{-4}$ }\\
		&& & \multicolumn{2}{c}{infeas. $10^{-5}$ }& &\multicolumn{2}{c}{infeas. $10^{-5}$ } \\
		\hline
	\end{tabular} \label{tab:nDecConstrOPF}\\[.01em]
	\smaller $^*$terminated because  max. iterations reached.
\end{table}

 Both, ADMM and the open-source solver \texttt{OSQP} are not able to solve the HVAC problems to a sufficient accuracy with a reasonable number of iterations (4,000 for \texttt{OSQP} and 5,000 for ADMM).
This indicates that these problems are rather challenging, which might be due to the large number of inequality constraints coming from the temperature bounds.
The computation time of ADMM, however, remains reasonable even for 5,000 iterations since we use warm-started \texttt{OSQP} as a very fast local solver parallelized via multi-threading.
In case less-strict termination tolerances are required, ADMM might become as fast as primal decomposition, overall.
Both primal decomposition schemes require far less iterations.
\texttt{Ipopt} requires more than 100 iterations, which is unusual for interior-point methods.
A possible explanation  is that interior-point methods typically choose the smallest stepsize such that no inequality constraint is violated (fraction-to-boundary rule) \cite[Eq 19.9]{Nocedal2006}. This can lead to  a slow progress since in this case only small steps are taken.
Primal decomposition mitigates this, since each subproblem has its ``own'' stepsize when solving the subproblems.

For the OPF problems~\eqref{eq:OPF}, computation times are generally lower because of smaller problem dimensions, cf. \autoref{tab:nDecConstr}.
	The qualitative performance of all algorithms relative to each other is similar to what we have observed for the HVAC problems.

Internal timings for the AL formulation and different sizes $|\mathcal S|$ for the HVAC problem are shown in \autoref{tab:interTime}.
Here, the time spent in the coordination problem \eqref{eq:coordPob} and in the local solvers stays relatively constant for varying $|\mathcal S|$.
The time spent for sensitivity computation, however, increases significantly with $|\mathcal S|$.
One explanation for that is that the complexity in the backsolves for computing the Hessian via \eqref{eq:dpStarDy} increases with $O(n_y^3)$.
In case  the backsolves can be parallelized, for example if the subproblems themselves use multi-threading in a cluster environment, computation time can be reduced.

\begin{table}
	\centering
	\caption{Internal timing (\%) for the AL formulation and the HVAC problem, 30 cores.}
	\begin{tabular}{rccccc}
		\hline
		$|\mathcal S|$ 		&sensitivity eval. &   local sol. & coord. &line search & other   \\
		\hline
		300 & 68.10 & 6.66 & 6.10 & 17.84 & 1.30 \\
		180 & 41.58 & 19.86 & 9.02 & 26.89 & 2.65\\
		30& 4.37 & 6.69 & 9.35 & 79.29 &  0.30 \\
		\hline
	\end{tabular} \label{tab:interTime}
\end{table}

\subsubsection*{Feasibility and Optimality}
Reaching feasibility fast is often crucial in the context of infrastructure systems to ensure system stability.
Here, primal decomposition can shine.
For all case studies, one can observe a high degree of primal feasibility in 10-20 iterations.
ADMM requires far more iterations to reach a sufficient degree of feasibility, which is a known limitation of ADMM \cite{Boyd2011}.

\begin{table}[t]
	\centering
	\caption{Properties of  ADMM and Primal Decomposition.}
	\begin{tabular}{crlllllll}
		\hline			& & ADMM  & 		Primal Decomp.\\
		\hline
		Communication & forward& $\mathcal O( n_y)$ & $\mathcal O( (N_{ls} +1) n_y)$\\
		(\#floats step)& backward& $\mathcal O( n_y)$ &  $\mathcal O((n_y^i)^2 + n_y) $ \\
		\hline
		Computation & local& Convex QP & NLP \\
		& global & Convex QP & NLP +  lin. equations \\
		&&& with multiple rhs. \\
		\hline
		Conv. rate (max.) && Linear & (Superlinear)$^\#$ \\
		\hline
		Decentralization & &Decentralized$^*$ & Distributed \\
		\hline
	\end{tabular} \\[.05cm]
	\smaller  $^*$in the sense that decentralization of ADMM is straight-forward. \newline
	\smaller $^\#$to a solution of the barrier problem \eqref{eq:ValueFunBarr2}
	\label{tab:compADMMPrimDec}
\end{table}

\subsubsection*{Communication}
Primal decomposition (\autoref{alg:maProbSolve}) requires communication of the current iterate $y^k$ to all subsystems $\mathcal O(N_{ls} +1)$ times, where $N_{ls}$ is the number of line search steps.
For backward communication, both primal decomposition schemes  require the gradient $g_i \doteq\nabla_y \phi_i $  and the Hessian  $H_i \doteq \nabla_{yy}^2 \phi_i$ of the optimal value functions once in each outer iteration from the subsystems to the master.
The communication of the Hessian is the most expensive step, but  the $H_i$ are typically highly sparse and symmetric, which reduces its communication demand.
\autoref{fig:sparsityplots} exemplarily shows the sparsity pattern of the Hessian of the second subsystem $H_2=\nabla_{yy}^2 \phi_2 (y)$ for one HVAC and one OPF problem.
For OPF, $H_2$  is diagonal and thus very cheap to communicate.
For HVAC, $H_2$ is almost diagonal except for one dense block.
This comes from the  fact that the local Hessians $\{H_i\}_{i\in \mathcal S}$ have non-zero entries only for the coupling elements of $y$ to the $i$th subsystem  plus the non-zero entries in $H_i^{yy}$ cf. \eqref{eq:HessAugLag}.
We denote by $n_y^i$ the number of coupling variables between the $i$th subsystem and the master in the following.
For the OPF problems, the subsystems are coupled to the master only via one bus of active power exchange, which leads to a $(1,1)$-Hessian block and thus $H_2$ is diagonal.
Similarly for the HVAC problems, the subproblems are coupled to the master via power exchange occurring over $24$ time steps, which leads to a dense $(24,24)$-Hessian block.
Thus, each subsystem $i \in \mathcal S$ needs to communicate $\mathcal O(n_y^i(n_y^i +1)/2 + n_y +N_{ls}^k + 1)$ floats to the master taking  the Hessian's symmetry and communication of $g_i$ and $\phi_i(y^k)$ for each line search step into account.

ADMM requires to communicate $y^k$ from the master to all subproblems and the corresponding local equivalent $z_i^k$ from all subsystems to the master in each iteration \cite{Engelmann2022b}.
Hence, forward and backward communication is $\mathcal O(n_y)$.

\autoref{fig:Comm} shows the forward and backward communication for AL-based primal decomposition and ADMM for the HVAC problem with $|\mathcal S | = 30$, and for the OPF problem with $|\mathcal S | = 39$ for one subsystem.
One can see that for the HVAC problem, AL's backward communication is approximately twice the ADMM communication as we have to communicate the dense Hessian block.
However, since $n_i^y \ll n_y$, the difference remains small.
In later iterations, forward communication increases due to a higher number of line search steps.

For the OPF problem, the forward communication costs of primal decomposition and ADMM are similar since no line search is required.
Backward communication is  twice as high in primal decomposition because of the diagonal Hessian.

\begin{figure}
	\centering
	\includegraphics[width=0.4\linewidth]{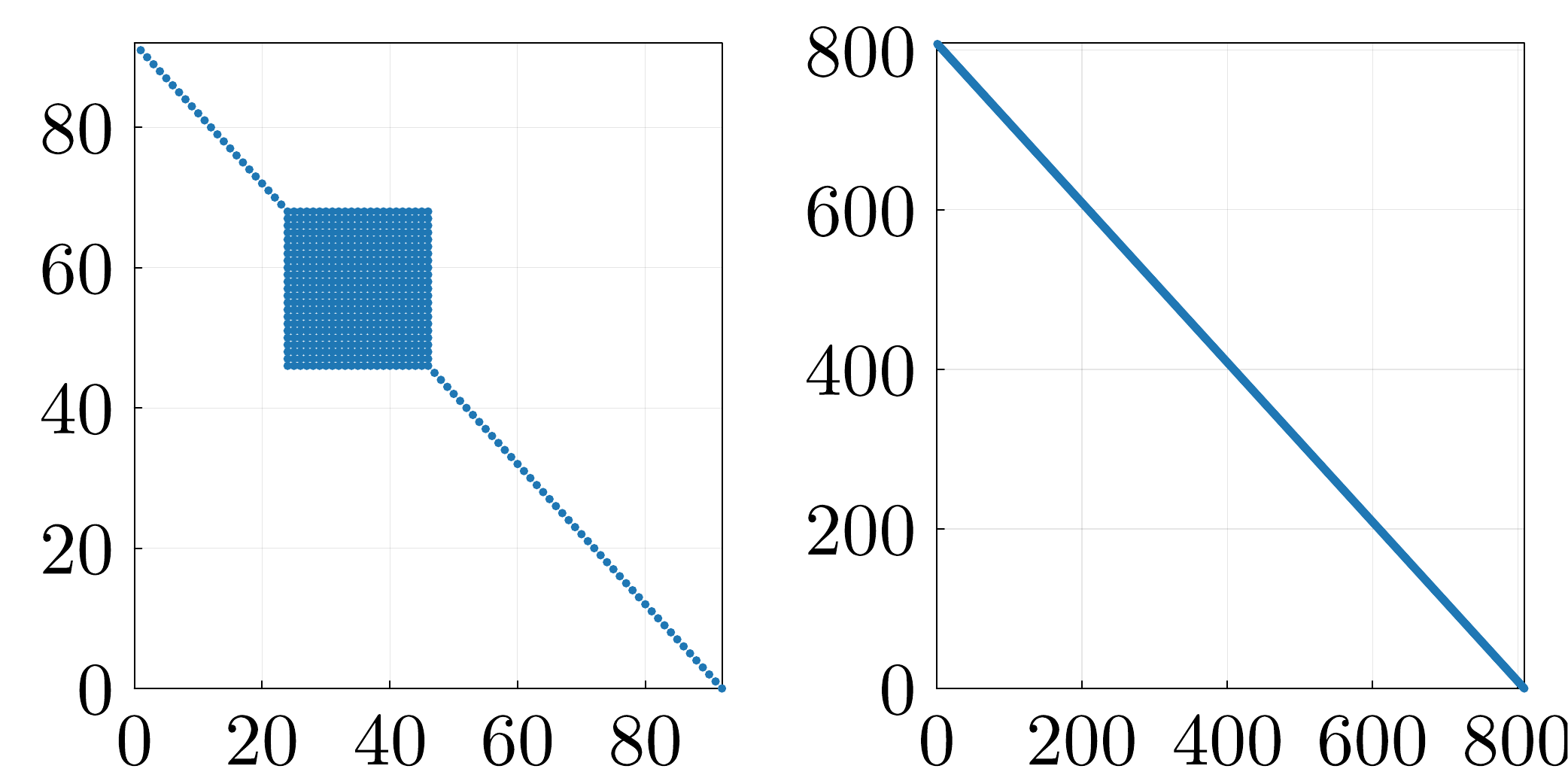}
	\caption{Sparsity patterns of $\nabla_{yy}^2 \phi_2 (y)$ for the HVAC problem with $|\mathcal S| = 4$  (left) and the OPF problem with $|\mathcal S| = 29$ (right).}
	\label{fig:sparsityplots}
\end{figure}
\begin{figure}
	\centering
   \begin{subfigure}[b]{.6\textwidth}
	\centering
\includegraphics[width=\linewidth]{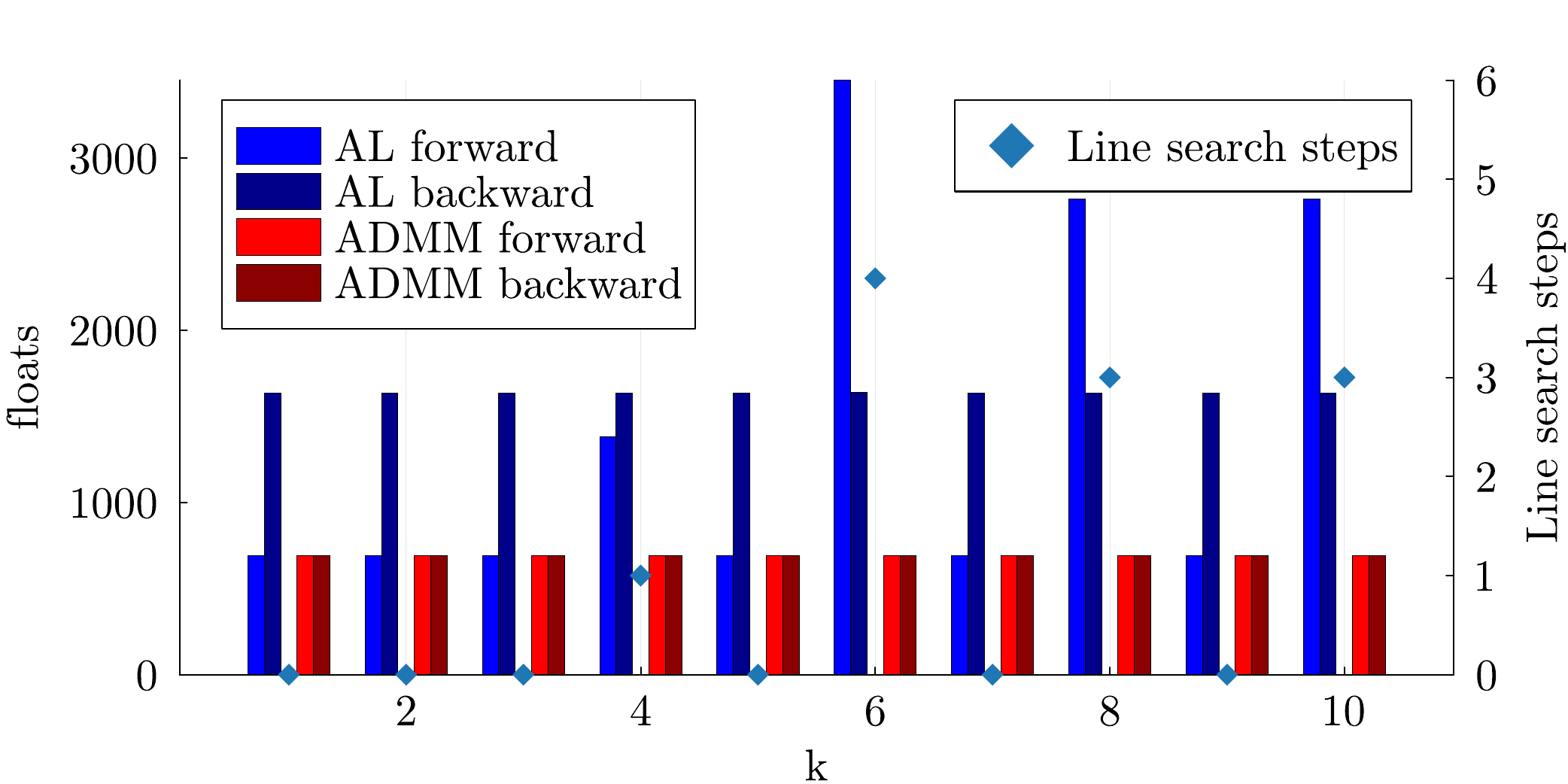}
\caption{ HVAC problem  with $|\mathcal S| = 30$.}
\label{fig:commplotshvac30}
\end{subfigure}
\hfill
\begin{subfigure}[b]{.6\textwidth}
	\centering
	\vspace{-.07cm}
\includegraphics[width=\linewidth]{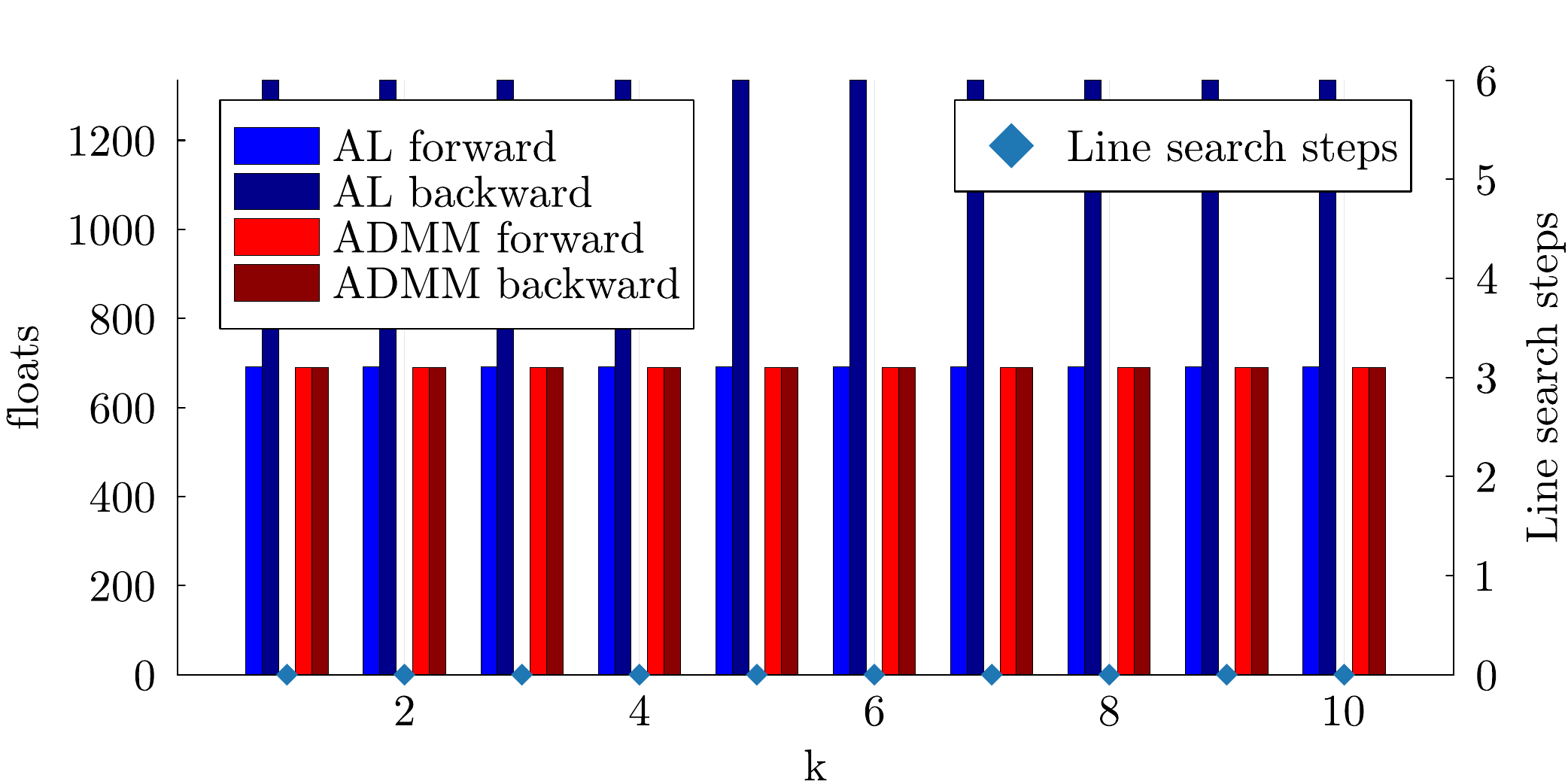}
\label{fig:commplotsdcopf29}
\caption{ OPF problem  with $|\mathcal S| = 29$.}
\end{subfigure}
\vspace{-.1cm}
\caption{Communication (\#floats) for one subsystem.}
\label{fig:Comm}
\end{figure}

\subsubsection*{Decentralization}
If decentralization is required, i.e. avoiding central coordination, ADMM is advantageous since decentralization is straight-forward.
Important properties of primal decomposition and ADMM are summarized in \autoref{tab:compADMMPrimDec}.

\section{Conclusion and Outlook}
We have presented two primal decomposition schemes to solve large-scale QPs for the operation of infrastructure networks.
The developed methods are proven to converge globally to the optimal solution.
Numerical experiments have demonstrated their potential for solving large-scale QPs in a  small number of iterations to a high degree of feasibility and optimality, which distinguishes them from classical distributed methods such as ADMM.
Moreover, we have shown that primal decomposition based on augmented Lagrangians has numerical benefits compared to the classical $\ell$1-formulation.

Future work will further improve implementation aspects of the developed primal decompositions schemes. Sparse backsolves or quasi-Newton Hessian approximations have the potential to greatly accelerate Hessian computation.

 \appendix

\section{Sensitivities  for   Augmented Lagrangians } \label{sec:QPsens}
Observe that for computing $\nabla_y \Phi_i^\delta$ and $\nabla_{yy}^2 \Phi_i^\delta$ in \eqref{eq:nablaYPhi} and~\eqref{eq:Hess}, the partial derivatives of the implicit function $F_i^\delta$ and $L_i$ are required.
Next, we derive these quantities for the two relaxed local problems \eqref{eq:subProbBarrQP2} and \eqref{eq:subProbBarrQP3}.

For \eqref{eq:subProbBarrQP2}, the Lagrangian (omitting  arguments) reads
\begin{align*}
	L_i^{\delta,\rho} \hspace{-.3em}\doteq&\hspace{-.0em}\frac{1}{2}  \hspace{-.3em}
	\begin{bmatrix}
		x_i \\  y \\z_i
	\end{bmatrix}^{\hspace{-.4em}\top}
	\hspace{-.5em}
	\begin{bmatrix}
		H_i^{xx} & H_i^{xy} &0 \\
		H_i^{xy\top} & \hspace{-.6em}H_i^{yy} + \rho I\hspace{-.6em}& -\rho I \\
		0 & -\rho I & \rho I
	\end{bmatrix}
	\hspace{-.4em}
	\begin{bmatrix}
		x_i \\ y \\z_i
	\end{bmatrix}
	\hspace{-.3em}
	+
	\hspace{-.3em}
	\begin{bmatrix}
		h_i^x \\ h_i^y + \lambda_i^k \\ -\lambda^k_i
	\end{bmatrix}^{\hspace{-.4em}\top}
	\hspace{-.5em}
	\begin{bmatrix}
		x_i \\  y \\z_i
	\end{bmatrix} \\
	&- \delta \mathds 1^\top \ln (s_i)
	+ \gamma_i^\top
	\left (
	\begin{bmatrix}
		A_i^x & A_i^y
	\end{bmatrix}
	\begin{bmatrix}
		x_i^\top & z_i^\top
	\end{bmatrix}^\top
	-b_i
	\right )\\
	& +
	\mu_i^\top
	(
	\begin{bmatrix}
		B_i^x & B_i^y
	\end{bmatrix}
	\begin{bmatrix}
		x_i^\top & z_i^\top
	\end{bmatrix}^\top + s_i -d_i
	).
\end{align*}
Hence, the local KKT conditions read
\begin{align*}
	T_i^{\delta,\rho}(q_i^\star,y)	\hspace{-.2em} \doteq 	\hspace{-.3em}
	\begin{bmatrix}
		H_i^{xx} x_i^\star + H_i^{xy} y + h_i^x + A_i^{x\top} \gamma_i^\star + B_i^{x\top} \mu_i^\star \\
		\rho( z_i^\star - y) - \lambda_i^k +A_i^{y\top} \gamma^\star_i + B_i^{y\top } \mu_i^\star \\
		- (S_i^\star)^{-1}\delta  \mathds 1 +  \mu_i^\star  \\
		\begin{bmatrix}
			A_i^x & A_i^y
		\end{bmatrix}
		\begin{bmatrix}
			x_i^{\star \top} & z_i^{ \star \top}
		\end{bmatrix}^\top -b_i \\
		\begin{bmatrix}
			B_i^x & B_i^y
		\end{bmatrix}
		\begin{bmatrix}
			x_i^{\star\top} & z_i^{ \star \top}
		\end{bmatrix}^\top + s_i^\star -d_i
	\end{bmatrix}
	\hspace{-.4em}
	=	\hspace{-.1em}0,
\end{align*}
where $q_i^\top \doteq \left [x_i^\top, z_i^\top, s_i^\top, \gamma_i^\top, \mu_i^\top\right ]$.
Moreover,
\begin{align}
	\nabla_{q_i} T_i^{\delta,\rho} (q_i,y)
	&=
	\begin{bmatrix}
		H_i^{xx} & 0 & 0 &A_i^{x \top } &  B_i^{x \top } \\
		0 &  \rho I & 0 & A_i^{y\top } & B_i^{y\top}\\
		0 & 0 & \hspace{-.3em} S_i^{-1}M_i \hspace{-.3em} & 0 & I \\
		A_i^x & A_i^y & 0 & 0 & 0  \\
		B_i^x &  B_i^y & I & 0 & 0
	\end{bmatrix},\\
	\nabla_{y} T_i^{\delta,\rho} (q_i,y) \label{eq:nablayT}
	&=
	\begin{bmatrix}
		H_i^{xy\top} &
		-\rho I &
		0 &
		0 &
		0
	\end{bmatrix}^\top,
\end{align}
where $M_i = \operatorname{diag}(\mu_i)$.
Moreover,  by \eqref{eq:nablaYPhi},
$
	\nabla_y 	\Phi_i^{\delta,\rho}(y) =  \nabla_y L_i^{\delta,\rho}(q_i^\star(y);y) =  (H_i^{yy} + \rho I )y  + H_i^{xy}x_i^\star + h_i^y + \lambda_i^k - \rho z_i^\star.
$
Furthermore, by \eqref{eq:Hess},
\begin{align}  \label{eq:HessAugLag}
	\nabla_{yy}^2 	\Phi_i^{\delta,\rho}(y) &\hspace{-.2em}= \hspace{-.2em}H_i^{yy} \hspace{-.2em}+\hspace{-.1em} \rho I \hspace{-.1em}+\hspace{-.1em} [H_i^{xy\top} -\hspace{-.2em}\rho I \;\;  0 \;\; 0 \;\; 0]\, \nabla_y q_i^\star (y),\hspace{-.3em}
\end{align}
where $q_i^\star (y)$ is computed by the system of linear equations \begin{align}\label{eq:impFunT}
	\nabla_{q_i} T_i^{\delta,\rho}(q_i^\star,y)\, \nabla_y q_i^\star (y)= -\nabla_{y} T_i^{\delta,\rho} (q^\star_i,y)
\end{align}

\subsubsection*{Precomputation for Hessian Evaluation}
Next, we show how to  precompute matrices to make \eqref{eq:impFunT} easier to solve, cf. \cite[Sec IV]{Pacaud2022a}.
We assume that $H_i^{xx}$ is invertible---if this is not the case, one can use the variant without precomputation.
Recall that by \eqref{eq:HessAugLag}, we need to compute $ H_i^{xy\top} \nabla_y x_i^\star  -\rho  \nabla_y z_i^\star$, where $( \nabla_y x_i^\star, \nabla_y z_i^\star)$ are given by \eqref{eq:impFunT}:
\begin{align} \label{eq:KKTsysAL}
	\begin{bmatrix}
		H_i^{xx} & 0 & 0 &A_i^{x \top } &  B_i^{x \top } \\
		0 &  \rho I & 0 & A_i^{y\top } & B_i^{y\top}\\
		0 & 0 & \hspace{-.8em} S_i^{\star -1}M_i^\star\hspace{-.8em} & 0 & I \\
		A_i^x & A_i^y & 0 & 0 & 0  \\
		B_i^x &  B_i^y & I & 0 & 0
	\end{bmatrix}
\hspace{-.5em}
	\begin{bmatrix}
		\nabla_y x^\star_i \\
		\nabla_y z^\star_i \\
		\nabla_y s^\star_i \\
		\nabla_y \gamma^\star_i \\
		\nabla_y \mu^\star_i \\
	\end{bmatrix}
\hspace{-.4em}
	=
	\hspace{-.4em}
	\begin{bmatrix}
		-H_i^{xy} \\
		\rho I \\
		0 \\
		0 \\
		0
	\end{bmatrix} \hspace{-.2em}.
\end{align}
By the third block-row, we have that $\nabla_y s^\star_i = - M^{\star-1}_i S_i^\star \nabla_y \mu^\star_i $.
This yields
\begin{align*}
	\begin{bmatrix}
		H_i^{xx} & 0  &A_i^{x \top } &  B_i^{x \top } \\
		0 &  \rho I  & A_i^{y\top } & B_i^{y\top}\\
		A_i^x & A_i^y  & 0 & 0  \\
		B_i^x &  B_i^y  & 0 &  -M^{\star-1}_i S_i^\star
	\end{bmatrix}
	\begin{bmatrix}
		\nabla_y x^\star_i \\
		\nabla_y z^\star_i \\
		\nabla_y \gamma^\star_i \\
		\nabla_y \mu^\star_i \\
	\end{bmatrix}
	=
	\begin{bmatrix}
		-H_i^{xy} \\
		\rho I \\
		0 \\
		0
	\end{bmatrix}.
\end{align*}
Since $H_i^{xx}$ is invertible, we have
\begin{align} \label{eq:preSchurComp}
	\begin{bmatrix}
		\nabla_y x^\star_i \\
		\nabla_y z^\star_i \\
	\end{bmatrix}
\hspace{-.4em}
	=
	\hspace{-.4em}
\underbrace{	\begin{bmatrix}
		H_i^{xx} \hspace{-1.0em} & 0   \\
		0 &  \rho I
\end{bmatrix}^{-1}}_{\doteq P_i^{-1}}
\hspace{-.3em}
	\Bigg [
	\hspace{-.3em}
	\begin{bmatrix}
		-H_i^{xy} \\
		\rho I
	\end{bmatrix}
\hspace{-.4em}
	-
	\hspace{-.4em}
	\underbrace{	\begin{bmatrix}
		A_i^x \hspace{-.3em}& A_i^y   \\
		B_i^x \hspace{-.3em}&  B_i^y
	\end{bmatrix}^{\hspace{-.3em}\top}}_{\doteq K_i^\top}
\hspace{-.4em}
	\begin{bmatrix}
		\nabla_y \gamma^\star_i \\
		\nabla_y \mu^\star_i \\
	\end{bmatrix}
\hspace{-.3em}
	\Bigg]
	\hspace{-.2em}.
\end{align}
Employing the Schur complement with respect to the first two block rows yields
\begin{align}
	\Bigg [\hspace{-.1em}
	K_i \notag
P_i^{-1} \hspace{-.1em}
	K_i^\top
\hspace{-.4em}	+\hspace{-.2em}
\underbrace{	\begin{bmatrix}
		0 & 0  \\
		0 & \hspace{-.7em} M^{\star-1}_i S_i^\star\hspace{-.1em}
	\end{bmatrix}}_{\doteq W_i}\hspace{-.4em}
	\Bigg ]
	\hspace{-.4em}
	\begin{bmatrix}
		\nabla_y \gamma^\star_i \\
		\nabla_y \mu^\star_i \\
	\end{bmatrix}
\hspace{-.4em}	=\hspace{-.4em}
	\underbrace{	\begin{bmatrix}
			-	A_i^{x  } H_i^{xx-1} H_i^{xy} \hspace{-.1em}+ \hspace{-.1em} A_i^{y  } \\
			B_i^{x } H_i^{xx-1} H_i^{xy}  \hspace{-.1em}+ \hspace{-.1em}B_i^{y} \\
	\end{bmatrix}}
	_{\doteq R_i}\hspace{-.3em}. \notag
\end{align}
Observe that $H_i^{xx-1}, K_i$ and $R_i$ can be precomputed.
Moreover, the above system of linear equations has significantly less decision variables compared to \eqref{eq:impFunT} and is in addition positive definite under Assumption~\autoref{ass:stdAss}.
This allows to use the Cholesky or Bunch-Kaufmann~(LDL) factorization  instead of~LU.
\Aee{
\subsubsection*{Precomputation for the Local Solvers}
Similar to the above, we can use precomputation for speeding up the solution of the local optimization problems \eqref{eq:subProbBarrQP2} in case interior-point methods are used.
Here, we need to compute Newton steps $	\nabla_{q_i} T_i^{\delta,\rho} (q_i,y) \,\Delta q_i = - 	T_i^{\delta,\rho}(q_i,y)$, i.e.
\begin{align} \label{eq:KKTsysAL2}
	&\begin{bmatrix}
		H_i^{xx} & 0  &A_i^{x \top } &  B_i^{x \top } \\
		0 &  \rho I  & A_i^{y\top } & B_i^{y\top}\\
		A_i^x & A_i^y  & 0 & 0  \\
		B_i^x &  B_i^y  & 0 &  \hspace{-.8em} -M^{-1}_i S_i
	\end{bmatrix}
\hspace{-.4em}
	\begin{bmatrix}
		\Delta x_i \\
		\Delta z_i \\
		\Delta \gamma_i \\
		\Delta \mu_i \\
	\end{bmatrix}
	= \\
&	- \hspace{-.4em}
	\begin{bmatrix}
		H_i^{xx} x_i + H_i^{xy} y + h_i^x + A_i^{x\top} \gamma_i + B_i^{x\top} \mu_i \\
		\rho( z_i - y) - \lambda_i^k +A_i^{y\top} \gamma_i + B_i^{y\top } \mu_i \\
		\begin{bmatrix}
			A_i^x \hspace{-.4em}& A_i^y
		\end{bmatrix}
		\begin{bmatrix}
			x_i^{ \top} \hspace{-.4em}& z_i^{  \top}
		\end{bmatrix}^\top -b_i \\
		\begin{bmatrix}
			B_i^x  \hspace{-.4em}& B_i^y
		\end{bmatrix}
	 \hspace{-.2em}
		\begin{bmatrix}
			x_i^{\top}  \hspace{-.4em}& z_i^{ \top}
		\end{bmatrix}^\top  \hspace{-.6em} + \hspace{-.1em} s_i  \hspace{-.1em}- \hspace{-.1em}d_i  \hspace{-.1em}+ \hspace{-.1em}M^{-1}_i(\delta \mathds{1} - S_i\mu_i)
	\end{bmatrix} \notag
	 \hspace{-.4em}\doteq \hspace{-.4em}
	\begin{bmatrix}
		g_1\\ g_2\\ h_1 \\h_2
	\end{bmatrix}  \hspace{-.3em},
\end{align}
where we have eliminate the third block-row via  $\Delta s_i = - M^{-1}_i (S_i \Delta \mu_i -\delta \mathds 1 + S_i \mu_i)$.
Solving for the first two block-rows yields
\begin{align*}
	\begin{bmatrix}
		\Delta x_i \\
		\Delta z_i \\
	\end{bmatrix}
	=
	P_i^{-1}
	\left (
	\begin{bmatrix}
		g_1 \\
		g_2
	\end{bmatrix}
	-
K_i^\top
	\begin{bmatrix}
		\Delta \gamma_i \\
		\Delta \mu_i \\
	\end{bmatrix}
	\right),
\end{align*}
where $P_i^{-1}$ and $K_i$ are from \eqref{eq:preSchurComp}.
Again, employing the Schur-complement with respect to the first two block rows yields
\begin{align*}
	&\left  (
K_i
P_i^{-1}
K_i^\top
+
W_i
\right  )
\hspace{-.4em}
	\begin{bmatrix}
		\Delta \gamma_i \\
		\Delta \mu_i \\
	\end{bmatrix}  =
	K_i
P_i^{-1}
	\begin{bmatrix}
		g_1 \\
		g_2 \\
	\end{bmatrix}
	-
	\begin{bmatrix}
		h_1 \\
		h_2
	\end{bmatrix}.
\end{align*}
Observe that the above is again  a small system of linear equations with positive definite coefficient matrix, which allows for using Cholesky factorization and  precomputed $H_i^{xx-1}$ and~$K_i$.
}

\section*{Sensitivities for the $\ell1$ Formulation}
The Lagrangian to \eqref{eq:subProbBarrQP3} reads
\begin{align*}
	L_i &=  	\frac{1}{2}
	\begin{bmatrix}
		x_i \\ z_i
	\end{bmatrix}^\top
	\begin{bmatrix}
		H_i^{xx} & H_i^{xy} \\
		H_i^{xy\top} & H_i^{yy}
	\end{bmatrix}
	\begin{bmatrix}
		x_i \\ z_i
	\end{bmatrix}
	+
	\begin{bmatrix}
		h_i^x \\ h_i^y
	\end{bmatrix}^\top
	\begin{bmatrix}
		x_i \\ z_i
	\end{bmatrix} \\
	&	+ \bar \lambda   \mathds 1^\top (v_i+w_i)
	-\delta  (\mathds 1^\top\ln(s_i) + \mathds 1^\top\ln(v_i) +\mathds 1^\top\ln(w_i) ) 	\\
	&	 + \chi_i^\top \hspace{-.1em}(y-z_i - v_i + w_i)
	+ \gamma_i^\top \hspace{-.3em}\left ( \hspace{-.1em}
	\begin{bmatrix}
		A_i^x & A_i^y
	\end{bmatrix}
	\hspace{-.2em}
	\begin{bmatrix}
		x_i^\top & z_i^\top
	\end{bmatrix}^\top
	\hspace{-.5em}-b_i 	\hspace{-.1em}\right ) \\
	&	+
	\mu_i^\top
	(
	\begin{bmatrix}
		B_i^x & B_i^y
	\end{bmatrix}
	\hspace{-.1em}
	\begin{bmatrix}
		x_i^\top & z_i^\top
	\end{bmatrix}^\top + s_i -d_i
	).
\end{align*}
Hence, the KKT conditions require
\begin{align*}
		T_i^{\delta,\bar \lambda}(u_i^\star,y) \doteq 
	\begin{bmatrix}
		H_i^{xx} x_i^\star + H_i^{xy} z_i^\star + h_i^x + A_i^{x\top} \gamma_i^\star + B_i^{x\top} \mu_i^\star \\
		H_i^{yy}z_i^\star + H_i^{xy\top }x_i^\star + h_i^y -\chi_i^\star  +A_i^{y\top} \gamma^\star_i + B_i^{y\top } \mu_i^\star \\
		-S_i^{-1}\delta  \mathds 1 + \mu_i^\star  \\
		-\delta V_i^{-1}  \mathds 1 + (\bar \lambda \mathds 1 - \chi_i^\star)  \\
		-\delta W_i^{-1}  \mathds 1 + (\bar \lambda \mathds 1 + \chi_i^\star)  \\
		\begin{bmatrix}
			A_i^x & A_i^y
		\end{bmatrix}
		\begin{bmatrix}
			x_i^{\star \top} & z_i^{ \star \top}
		\end{bmatrix}^\top -b_i \\
		\begin{bmatrix}
			B_i^x & B_i^y
		\end{bmatrix}
		\begin{bmatrix}
			x_i^{\star\top} & z_i^{ \star \top}
		\end{bmatrix}^\top + s_i^\star -d_i \\
		y -z_i^\star - v_i^\star + w_i^\star
	\end{bmatrix}
	\overset{	!}{=}0,
\end{align*}
where $u_i^\top \doteq \left [x_i^\top, z_i^\top, s_i^\top,v_i^\top,w_i^\top, \gamma_i^\top, \mu_i^\top ,\chi_i^\top \right ]$.
Thus,
\begin{align} \label{eq:KKTsl12}
	\nabla_{u_i} T_i^{\delta,\bar \lambda }(u_i,y)
	=  
	\begin{bmatrix}
		H_i^{xx}  \hspace{-.6em}& H_i^{xy} \hspace{-.6em} & 0 & 0 & 0 & A_i^{x \top } \hspace{-.8em} &  B_i^{x \top }\hspace{-.6em} & 0 \\
		H_i^{xy\top}  \hspace{-.6em}\hspace{-.6em}& H_i^{yy}  \hspace{-.6em}& 0 & 0 & 0 &  A_i^{y\top }  \hspace{-.8em}& B_i^{y\top}\hspace{-.6em} & -I\\
		0 & 0 & \hspace{-.9em} S_i^{-1}M_i \hspace{-.5em} & 0 & 0 &  0 & I & 0\\
		0 & 0 & 0 & \hspace{-.2em}V_i^{-1}(\bar \lambda I-X_i) \hspace{-1em}& 0 &  0 & 0 & -I\\
		0 & 0 & 0 & 0 &\hspace{-1em}W_i^{-1}(\bar \lambda I + X_i)\hspace{-1em} &  0 & 0 & I\\
		A_i^x & A_i^y & 0 & 0 & 0  & 0 & 0 & 0\\
		B_i^x &  B_i^y & I & 0 & 0 & 0 & 0 & 0 \\
		0 & -I & 0 & -I & I &  0 & 0 & 0
	\end{bmatrix} , 
\end{align}
where  $V_i = \operatorname{diag}(v_i)$,   $W_i = \operatorname{diag}(w_i)$, and $X_i = \operatorname{diag}(\chi_i)$.
Moreover,
\begin{align} \label{eq:nablaYT}
	\nabla_{y} T_i^{\delta,\bar \lambda }(u_i,y)
=[	0 \;\; 0 \;\; 0 \;\; 0 \;\; 0 \;\; 0 \;\; 0 \;\; I]^\top.
\end{align}
Furthermore, $\nabla_y 	\Phi_i^{\delta,\bar \lambda}(y)=\nabla_{y}L_i =  \chi_i$, $\nabla_{yy}L_i = 0$, and $  \nabla_{yu_i^\star}L_i= [0\;\; 0 \;\; 0 \;\; 0\;\; 0 \;\; 0\;\; 0 \;\; I]$.
Thus, by \eqref{eq:Hess},
\begin{align} \label{eq:Hess2QP4}
	\nabla_{yy} 	\Phi_i^{\delta,\bar \lambda}(y) = \nabla _y\chi_i^\star (y).
\end{align}

\section{Proof of Lemma~\autoref{lem:posDefHess}} \label{sec:ProofPosDef}
	First, we will show that $Z^\top(ZCZ^\top)^{-1}Z = Z^\top ZC^{-1}Z^\top Z$ for a regular, symmetric $C\in \mathbb{R}^{n\times n}$,  $Z \in \mathbb{R}^{m\times n}$ with $m<n$.
	Consider a re-ordered eigendecomposition $C = Q \Lambda Q^\top$ and partition $Q = [Q_{1} \;\;Q_{2}]$, $\Lambda_i = \operatorname{blkdiag}(\Lambda_{1},\Lambda_{2})$ such that $Q_{2}$ is a nullspace-basis of $Z$, i.e. $Z Q_{2}=0$.
	Hence, we have $ZCZ^\top = Z[Q_{1} \;\;Q_{2}]\operatorname{blkdiag}(\Lambda_{1},\Lambda_{2})[Q_{1} \;\;Q_{2}]^\top Z_i^\top = ZQ_{1} \Lambda_{1} Q_{1}Z^\top$ since $Z Q_{2}=0$.
	Thus, $Z^\top(ZCZ^\top )^{-1}Z= Z^\top Z Q_1 \Lambda_{1}^{-1}  Q_1^\top Z^\top Z $.
	Again, since  $Z Q_{2}=0$, by expansion, $ Z^\top Z [Q_1 \;\; Q_2]\operatorname{blkdiag}{( \Lambda_{1}^{-1},\Lambda_{2}^{-1})}  [Q_1 \;\; Q_2]^\top Z^\top Z =Z^\top ZC^{-1}Z^\top  Z$.

	Proof of a): By \eqref{eq:HessAugLag}, we need $\nabla_y q_i^\star (y)$ for computing $	\nabla_{yy} 	\Phi_i^{\delta,\rho}$, where  $\nabla_y q_i^\star (y)$ is defined by \eqref{eq:impFunT}.
	Define  $C_i \doteq \operatorname{blkdiag}(H_i^{xx},\rho I,S_i^{-1}M_i)),$ $D_i\doteq 	\begin{bmatrix}
		A_i^x & A_i^y & 0  \\
		B_i^x &  B_i^y & I
	\end{bmatrix}$, and
	$E_i \doteq [H_i^{xy\top} -\hspace{-.2em}\rho I \;\;\;  0]^\top.$
	Consider \eqref{eq:nablayT} and parametrize $(\nabla_y x_i^{\star\top},\nabla_y z_i^{\star\top},\nabla_y s_i^{\star\top})^\top \doteq Z_iP_i$, where $Z_i$ is a nullspace matrix to $D_i$, i.e., the columns of $Z_i$ form an orthogonal basis of the nullspace of $D_i$ and $P_i \in \mathbb{R}^{(n_{xi} + n_y -\operatorname{nr}(A_i^x))\times n_y}$ is an auxiliary matrix.
	Using the above parametrization in \eqref{eq:impFunT} and multiplying with $Z_i^\top$  yields
	$Z_i^\top C_iZ_i P_i= -Z_i^\top E_i$ by $Z_i^\top D_i^\top =0$.
	Since $s_i,\mu_i >0$ and Assumption~1 holds, we have $ C_i \succ 0$ and thus $Z_i^\top C_iZ_i$ is invertible by  full rank of $Z_i$.
	Hence, by~\eqref{eq:HessAugLag} and the above derivation,	$\nabla_{yy} 	\Phi_i^{\delta,\rho}(y) = H_i^{yy} + \rho I - E_i^\top\,Z_i (Z_i^\top C_iZ_i)^{-1}Z_i^\top E_i= H_i^{yy} + \rho I - E_i^\top\,Z_i Z_i^\top C_i^{-1}Z_iZ_i^\top E_i.$
	Notice that $Z_iZ_i^\top$ is a diagonal matrix with $\operatorname{rank}(Z_i)$ ones and $\operatorname{dim}(C_i) - \operatorname{rank}(Z_i)$ zeros.
	Hence, since $C_i$ is positive definite, it suffices to show that $\nabla_{yy} 	\Phi_i^{\delta,\rho}(y) \succ 0$ for the worst case, i.e. $Z_iZ_i^\top =I$ (no constraints).
	Thus, $\nabla_{yy} 	\Phi_i^{\delta,\rho}(y) = H_i^{yy} - H_i^{xy\top}(H_i^{xx})^{-1}H_i^{xy} \succ 0 $ by the definition of $E_i, C_i$, by Assumption~\ref{ass:stdAss}~a) and the Schur-complement Lemma \cite[A.14]{Boyd2004}.

	Proof of b): By \eqref{eq:Hess2QP4}, we need to show that $\nabla _y\chi_i^\star (y) \succ 0$, which can be computed by the system of linear equations  \eqref{eq:KKTsl12}, \eqref{eq:nablaYT}.
	Define $F_i = \operatorname{blkdiag}\bigg  (
	\begin{bmatrix}
		H_i^{xx} & H_i^{xy} \\
		H_i^{xy\top} & H_i^{yy}
	\end{bmatrix},
	S_i^{-1}M_i, V_i^{-1}(\bar \lambda -X_i),W_i^{-1}(\bar \lambda +X_i)
	\bigg  )$ and
	$
	G_i\doteq
	\begin{bmatrix}
		A_i^x & A_i^y & 0 & 0 & 0  & \\
		B_i^x &  B_i^y & I & 0 & 0 & \\
		0 & -I & 0 & -I & I &
	\end{bmatrix}.
	$
	By Assumption~\ref{ass:stdAss}, $s_i,v_i,w_i,\mu_i > 0$, and $\bar \lambda > \max_j{\left |[\chi_i]_j \right |}$, we have that $F_i \succ 0$.
	Hence, $\left (\nabla_y x_i^{\star\top},\nabla_y z_i^{\star\top}, \nabla_y s_i^{\star\top},\nabla_y v_i^{\star\top},\nabla_y w_i^{\star\top} \right )^\top = - F_i^{-1}G_i(\nabla_y \gamma_i^{\star\top},\nabla_y \mu_i^{\star\top}, \nabla_y \chi_i^{\star\top} )$.
	Thus, $G_i^\top F_i^{-1}G_i(\nabla_y \gamma_i^{\star\top},\nabla_y \mu_i^{\star\top}, \nabla_y \chi_i^{\star\top} ) = [0 \;\;0\;\; I]^\top $. Since $F_i^{-1}\succ 0$ and by full rank of $G_i$ from Assumption~\ref{ass:stdAss}, $G_i^\top F_i^{-1}G_i \succ 0$ and thus  $(\nabla_y \gamma_i^{\star\top},\nabla_y \mu_i^{\star\top}, \nabla_y \chi_i^{\star\top} ) = (G_i^\top F_i^{-1}G_i)^{-1}[0 \;\;0\;\; I]^\top $.
	Since $(G_i^\top F_i^{-1}G_i)^{-1} \succ 0$, all leading principle minors of this matrix must be positive definite by Sylvester's criterion \cite[Col 7.1.5]{Horn2013}.
	By variable reordering, the assertion follows.

\section{Solution of \eqref{eq:sepQP} via ADMM} \label{sec:ADMM}
We derive a distributed ADMM version for \eqref{eq:sepQP} as a baseline for numerical comparison.
Consider \eqref{eq:sepQP}, introduce auxiliary variables $z_i \in \mathbb{R}^{n_y}$  and  consensus constraints $y=z_1=\dots=z_S \;|\; \lambda_1,\dots,\lambda_S$.
This yields
\begin{subequations}\label{eq:sepQPADMM}
	\begin{align}
		\min_{x,y,z}\sum_{i \in \mathcal S}
		\frac{1}{2}
		\begin{bmatrix}
			x_i \\ z_i
		\end{bmatrix}^{	\hspace{-.2em}\top}
		\hspace{-.3em}
		\begin{bmatrix}
			H_i^{xx} & H_i^{xy} \\
			H_i^{xy\top}	\hspace{-.4em} & H_i^{yy}
		\end{bmatrix}
		\hspace{-.2em}
		\begin{bmatrix}
			x_i \\ z_i
		\end{bmatrix}&
	\hspace{-.2em}
		+ 	\hspace{-.2em}
		\begin{bmatrix}
			h_i^x \\ h_i^y
		\end{bmatrix}^{	\hspace{-.2em}\top}
		\hspace{-.2em}
		\begin{bmatrix}
			x_i \\ z_i
		\end{bmatrix} \label{eq:objADMM}
		\\
		\text{subject to } \;
		\begin{bmatrix}
			A_i^x & A_i^y
		\end{bmatrix}
		\begin{bmatrix}
			x_i^\top & z_i^\top
		\end{bmatrix}^\top - b_i &= 0,\;\; i \in \mathcal S, \\
		\begin{bmatrix}
			B_i^x & B_i^y
		\end{bmatrix}
		\begin{bmatrix}
			x_i^\top & z_i^\top
		\end{bmatrix}^\top -d_i&\leq 0, \;\;i \in \mathcal S, \label{eq:ineqADMM}\\
		A^y y -b^y = 0, \quad
		B^y y -d^y &\leq  0\label{eq:ineqyADMM}\\
		z_i &= y, \;\; i \in \mathcal{S}.
	\end{align}
\end{subequations}
The augmented Lagrangian with respect to $y-z_i=0$  reads
\begin{align*}
	L^\rho =  \sum_{i \in \mathcal S} \phi_i^x(x_i,z_i) + \phi^y(y)   + \lambda_i^{k\top}(y-z_i) + \frac{\rho}{2} \|y-z_i\|^2,
\end{align*}
where $\phi_i^x$ are defined by \eqref{eq:objADMM}-\eqref{eq:ineqADMM} and $\phi^y$ is the indicator function for \eqref{eq:ineqyADMM}.
Minimizing $L^\rho$ w.r.t. $(x_i,z_i)$ for fixed $(y^k,\lambda_i^k)$ yields for all $i \in \mathcal S$
\begin{align} \label{eq:ADMMlocMin}
	(x_i^{k+1},z_i^{k+1}) = 
&	\arg \min_{x_i,z_i} 
	\frac{1}{2}
	\begin{bmatrix}
		x_i \\ z_i
	\end{bmatrix}^{\hspace{-.2em}\top}
	\begin{bmatrix} \notag
		H_i^{xx} & H_i^{xy} \\
		H_i^{xy\top}\hspace{-.3em} & \hspace{-.5em}H_i^{yy} + \rho I
	\end{bmatrix}
	\begin{bmatrix}
		x_i \\ z_i
	\end{bmatrix}
	+
	\begin{bmatrix}
		h_i^x \\ h_i^y \hspace{-.2em}- \hspace{-.2em}\lambda_i^k\hspace{-.2em} -\hspace{-.2em} \rho y^k
	\end{bmatrix}^{\hspace{-.2em}\top}
	\begin{bmatrix}
		x_i \\ z_i
	\end{bmatrix}
	\\
	&
\begin{aligned}
		\text{subject to } \quad
	\begin{bmatrix}
		A_i^x & A_i^y
	\end{bmatrix}
	\begin{bmatrix}
		x_i^\top & z_i^\top
	\end{bmatrix}^\top - b_i &= 0, \\
	\begin{bmatrix}
		B_i^x & B_i^y
	\end{bmatrix}
	\begin{bmatrix}
		x_i^\top & z_i^\top
	\end{bmatrix}^\top -d_i&\leq 0.
\end{aligned}
\end{align}
Minimising $L^\rho$ w.r.t. $y$ for fixed  $(x_i^{k+1},z_i^{k+1},\lambda_i^k)$ yields
\begin{align} \label{eqADMMcons}
	y^{k+1} = \arg \min_{y}\;\sum_{i \in \mathcal S} \frac{\rho}{2}  \;y^\top  y &+(\lambda_i^k -\rho z_i^{k+1})^\top y
	\\
	\text{subject to } \quad
	A^y y -b^y &= 0, \quad B^y y -d^y \leq  0.\notag
\end{align}
Finally, the Lagrange multiplier update reads
\begin{align}
	\lambda_i^{k+1} = \lambda_i^k + \rho(y^{k+1}-z_i^{k+1}), \qquad i \in \mathcal{S}. \label{eq:ADMMlamUp}
\end{align}
The update rules \eqref{eq:ADMMlocMin}-\eqref{eq:ADMMlamUp} define the ADMM iterations.
Note that \eqref{eq:ADMMlocMin} and \eqref{eq:ADMMlamUp} can be executed locally for all $i \in \mathcal S$, whereas \eqref{eqADMMcons} defines the global coordination step.

\printbibliography

\end{document}